\documentclass[11pt]{article}

\usepackage[margin=1in]{geometry}
\usepackage{amsmath,amssymb,amsthm,amsxtra}
\allowdisplaybreaks[2]
\usepackage{mathrsfs, mathtools}
\usepackage{bbm}
\usepackage{bm}
\usepackage[shortlabels]{enumitem}
\usepackage{microtype}
\microtypesetup{expansion=false}
\usepackage[T1]{fontenc}
\usepackage{xcolor}
\usepackage{comment}
\usepackage{float}
\usepackage{hyperref}
\usepackage{graphicx}
\usepackage{algorithm}
\usepackage[noend]{algpseudocode}
\makeatletter
\renewcommand{\theHALG@line}{\thealgorithm.\arabic{ALG@line}}
\makeatother

\hypersetup{hidelinks}
\numberwithin{equation}{section}

\theoremstyle{plain}
\newtheorem{theorem}{Theorem}[section]
\newtheorem{lemma}[theorem]{Lemma}

\theoremstyle{definition}
\newtheorem{definition}[theorem]{Definition}

\newtheorem{fact}[theorem]{Fact}

\theoremstyle{remark}

\newcommand{\R}{\mathbbm{R}}

\newcommand{\E}{\operatorname{\mathbbm{E}}}

\newcommand{\ip}[2]{\left\langle #1, #2\right\rangle}

\makeatletter
\DeclareRobustCommand\Equiv{\mathrel{%
 \mathchoice
 {\Equiv@\textfont\displaystyle{.45}}
 {\Equiv@\textfont\textstyle{.45}}
 {\Equiv@\scriptfont\scriptstyle{.5}}
 {\Equiv@\scriptscriptfont\scriptscriptstyle{.55}}
}}
\newcommand{\Equiv@}[3]{%
 \rlap{\raisebox{#3\fontdimen5#12}{$\m@th#2 = $}}%
 \raisebox{-#3\fontdimen5#12}{$\m@th#2 = $}%
}
\makeatother

\newcommand{\benm}{\begin{enumerate}[leftmargin=*]}
\newcommand{\eenm}{\end{enumerate}}

\definecolor{forestgreen}{rgb}{0.13, 0.55, 0.13}

\newcommand{\Id}{\mathrm{Id}}
\newcommand{\kmsp}{KMS\(^{\prime}\)\space}
\newcommand{\kms}{KMS\space}

\algnewcommand{\Input}{\item[\textbf{Input:}]}
\algnewcommand{\Output}{\item[\textbf{Output:}]}

\title{Improved SDP Coloring of 3-Colorable Graphs from Recursive Gaussian Certificates}

\author{
  Ijay Narang\thanks{Georgia Institute of Technology,  School of Computer Science. \texttt{inarang3@gatech.edu}.}
  \and
  Yukai Tang\thanks{Princeton University, Operations Research and Financial Engineering. \texttt{yt3846@princeton.edu}.}
}
\date{}

\begin{document}
\maketitle

\begin{abstract}
We give a randomized polynomial-time algorithm that, for every fixed $\varepsilon > 0$, colors every \(3\)-colorable \(n\)-vertex graph using
\(O\bigl(n^{(13-\sqrt{97})/18+\varepsilon}\bigr) \approx O\bigl(n^{0.17506+\varepsilon}\bigr)\) colors, improving upon the previous best bound of \(O(n^{0.19539})\) from Bansal, Huang, and Lee.

 Our improvement comes from analyzing higher-level neighborhoods through a recursive description of failure in Gaussian SDP rounding. If the rounding returns too small an independent set, it produces local Gaussian certificates at every vertex of a nonempty induced subgraph. We propagate these certificates along walks to higher-level neighborhoods by defining a recursive certificate structure and
 proving a strengthened cover-composition lemma, which refines the one of Arora, Chlamt{\'a}{\v{c}}, and Charikar. We then construct a bounded potential function that increases by a fixed positive amount at every propagation step, yielding a contradiction. Consequently, the rounding must produce a sufficiently large independent set.

\end{abstract}

\clearpage

\section{Introduction}
\label{sec:Intro}

In this paper, we study the task of coloring \(3\)-colorable graphs efficiently.
More precisely, given a graph promised to admit a \(3\)-coloring, we ask how few colors a polynomial-time algorithm can guarantee to use.  This problem is known to be hard: recovering a \(3\)-coloring from a $3$-colorable graph is NP-hard~\cite{GJS74}, and it remains $\mathrm{NP}$-hard even when the algorithm is allowed a fixed number of colors ~\cite{FeiMinzerWang2026,GS20,BKLM22}.

Thus, considerable attention has been devoted to designing polynomial-time algorithms that use as few colors as possible. Progress has come through two closely intertwined
lines of work: combinatorial and SDP-based. On the combinatorial side, Wigderson's
neighborhood argument gives a simple \(O(\sqrt n)\)-coloring
algorithm~\cite{Wig83}, and Berger and Rompel improved this guarantee to \(O(\sqrt{n/\log n})\)~\cite{BR90}. Blum subsequently introduced new combinatorial techniques based on higher-order neighborhood structure and obtained a \(\widetilde O(n^{3/8})\)-coloring algorithm~\cite{Blu94}, where $\widetilde O (\cdot)$ suppresses logarithmic factors.

On the SDP side, the first such result was by Karger,
Motwani, and Sudan, who combined Gaussian rounding of a vector-coloring SDP with Wigderson's combinatorial routine to obtain a \(\widetilde O(n^{1/4})\)-coloring~\cite{KMS98}. Blum and Karger then combined
the same SDP method with Blum's stronger dense-graph machinery, improving the
bound to $\widetilde O(n^{3/14}) =\widetilde O(n^{0.2143})$
~\cite{BK97}. Further advances in SDP rounding led Arora,
Chlamt{\'a}{\v{c}}, and Charikar to obtain a
\(\widetilde O(n^{0.2111})\)-coloring~\cite{ACC06}, and Chlamt{\'a}{\v{c}} to improve
this to \(\widetilde O(n^{0.2072})\)~\cite{Chl07}.

Continued progress has come from both combinatorial and SDP-based methods. Kawarabayashi and Thorup introduced the first substantial new dense-graph techniques since Blum, obtaining successively \(\widetilde O(n^{0.2049})\) and
\(\widetilde O(n^{0.19996})\) colors~\cite{KT12,KT17}.
Together with Yoneda, they later improved the bound to
\(\widetilde O(n^{0.19747})\)~\cite{KTY24}.
Most recently, Bansal, Huang, and Lee shifted the frontier back to the SDP side. Through a refined analysis of second- and third-order neighborhoods, they obtained an \(O(n^{0.19539})\)-coloring algorithm~\cite{BHL26}, which
is the best previously known guarantee.

Our main result is an improved polynomial-time coloring guarantee for
\(3\)-colorable graphs.

\begin{theorem}
\label{thm:main}
For every fixed \(\varepsilon>0\), there is a randomized polynomial-time
algorithm that, given a \(3\)-colorable graph on \(n\) vertices, with high probability outputs a
proper coloring using
\(O\bigl(n^{(13-\sqrt{97})/18+\varepsilon}\bigr) \approx O\bigl(n^{0.17506+\varepsilon}\bigr)\) colors.
\end{theorem}

Our work continues progress on the SDP side, and at a high level can be viewed as a direct extension of the work of \cite{ACC06, Chl07, BHL26}. These works enable more aggressive rounding schemes by examining the level $2$ and level $3$ neighborhoods. We build on this by analyzing arbitrarily high-level neighborhoods through a recursive mechanism. Before detailing our approach, we review the vector coloring framework and the rounding algorithms on which our work builds.

\paragraph{Notation.} For notation, given a graph $G = (V,E)$ and a subset of vertices $S\subseteq V$, define $G[S]$ to be the subgraph induced by $S$. Given a subgraph $G' \subseteq G$, let $N_{G'}(v)$ be the set of vertices neighboring $v$ in $G'$ and $N_{G'}^{(\ell)}(v)$ the set of endpoints of length-$\ell$ walks beginning at $v$ in $G'$. Walks may repeat vertices. We denote the maximum and minimum degrees of $G$ by $\Delta(G)$ and $\delta(G)$, respectively.

We denote the identity matrix by $\Id$ and the multivariate Gaussian with mean $\mu$ and covariance $\Sigma$ by $\mathcal N (\mu, \Sigma)$. Additionally, let \(\Phi\) and \(\phi\) denote the CDF and PDF of
a standard Gaussian random variable. We write $\Phi^c$ to be the complement of the CDF and $\Phi^c(s):=1-\Phi(s)=\Phi(-s)$. We define $\Phi^{-1}$ to be the inverse Gaussian CDF (quantile function), with $\Phi^{-1}(0)=-\infty$ and $\Phi^{-1}(1)=+\infty$. All norms are Euclidean, and all logarithms are natural.

For a positive integer \(m\), write \([m]:=\{1,\ldots,m\}\).
We use \(\langle\cdot,\cdot\rangle\) for the Euclidean inner product and
\(x^\perp:=\{z:\langle x,z\rangle=0\}\). For \(r\in\mathbb R\), let
\(r_+:=\max\{r,0\}\). 

\subsection{The Vector \texorpdfstring{$3$}{3}-Coloring Rounding Framework} \label{subsec:rel_work}
A vector $3$-coloring of a graph $G = (V,E)$ is defined as an assignment of unit vectors $\{v_i\}_{i \in V}$ to vertices satisfying $\langle v_i,v_j\rangle \le-\frac12$ for every $\{i,j\}\in E.$ We say $\{v_i\}_{i \in V}$ is a \emph{strict} vector $3$-coloring of $G$ if $\langle v_i,v_j\rangle=-\frac12$ for every $\{i,j\}\in E.$ Vector $3$-colorings can be computed to arbitrary accuracy in polynomial time by semidefinite programming. For simplicity, and as in \cite{ACC06,Chl07,BHL26}, we write the analysis for strict vector $3$-colorings.

At a high level, obtaining a coloring with $\widetilde{O}(s(n))$ colors reduces to finding independent sets of size $\Omega(n/s(n))$ in the relevant induced subgraphs. We formalize this reduction using Blum's notion of progress in Appendix~\ref{subsec:dense-sparse-balance}. Our goal is to produce a large independent set from a strict vector $3$-coloring. We first review the KMS algorithm~\cite{KMS98}.

\begin{algorithm}[H]

\caption{\kms Gaussian rounding}
\label{alg:kms}

\begin{algorithmic}[1]

\Input A graph \(G=(V,E)\), a strict vector \(3\)-coloring
\(\{v_i\}_{i\in V}\), and a threshold \(t>0\).

\Output An independent set \(I\subseteq V\).

\State Sample a standard Gaussian vector
\(\gamma\sim\mathcal{N}(0,\mathrm{Id})\).

\State Form the set \(S=\{i\in V:\langle\gamma,v_i\rangle\geq t\}\).

\State Let \(I=\{i\in S:N_G(i)\cap S=\varnothing\}\) be the set of isolated vertices in the induced subgraph \(G[S]\).

\State \Return \(I\).

\end{algorithmic}

\end{algorithm}

The analysis of the \kms algorithm proceeds as follows. The algorithm chooses the threshold \(t\) such that
\(\Phi^c(t) \approx \Delta^{-1/3}\). Consequently,
\(\E[|S|]\approx n\Delta^{-1/3}.\)
Moreover, with this choice of \(t\), the conditional probability that a fixed neighbor of a selected vertex also lies in \(S\) is small enough that a union bound over its neighborhood shows that a constant fraction of the vertices selected in \(S\) are isolated in the induced subgraph \(G[S]\). This yields an independent set of size
$\widetilde\Omega(n\Delta^{-1/3})$.

We next describe a variant of KMS algorithm formulated in \cite{ACC06}.

\begin{algorithm}[H]
\caption{\kmsp Gaussian rounding}
\label{alg:kms-prime}
\begin{algorithmic}[1]
\Input A graph \(G=(V,E)\), a strict vector \(3\)-coloring
\(\{v_i\}_{i\in V}\), and a threshold \(t>0\).
\Output An independent set \(I\subseteq V\).

\State Sample a standard Gaussian vector $\gamma \sim \mathcal{N}(0,\mathrm{Id})$.

\State Form the set \(S=\{i\in V:\langle\gamma,v_i\rangle\geq t\}\).

\State Compute an arbitrary maximal matching \(M\) in the induced subgraph
\(G[S]\).

\State \Return \(I=S\setminus V(M)\).
\end{algorithmic}
\end{algorithm}

The analysis of \kmsp~\cite{ACC06} improves on the original KMS rounding procedure by lowering the threshold \(t\) so that $\Phi^c(t) \approx \Delta^\frac{-1}{3(1+c)}$, thereby increasing the expected size of the random set \(S\) to $n \Delta^\frac{-1}{3(1+c)}$. If a substantial fraction of the vertices in \(S\) survive the maximal-matching cleanup, then \kmsp returns a larger independent set.

We now analyze what happens when most of the vertices selected in $S$ are removed. Since $\{v_i\}$ is a strict vector \(3\)-coloring, every edge \(\{i,j\}\in E\) admits the decomposition
\begin{align}
v_j&=-\frac12v_i+\frac{\sqrt3}{2}u_{j\mid i},
\qquad \|u_{j\mid i}\|=1,
\qquad \langle u_{j\mid i},v_i\rangle=0.
\label{eq:edge-decomposition}
\end{align}
The decomposition from the other endpoint $u_{i \mid j}$ is analogously defined.
Consequently,
\begin{align}
&\Pr\left[
    \exists j\in N_G(i):\langle\gamma,v_j\rangle\geq t
    \,\middle|\,\langle\gamma,v_i\rangle\geq t
\right]\notag\\
&\qquad\leq\Pr\left[
    \exists j\in N_G(i):\langle\gamma,u_{j\mid i}\rangle\geq\sqrt3\,t
    \,\middle|\,\langle\gamma,v_i\rangle\geq t
\right]\notag\\
&\qquad=\Pr\left[
    \exists j\in N_G(i):\langle\gamma,u_{j\mid i}\rangle\geq\sqrt3\,t
\right],
\end{align}
where the final equality follows from \(u_{j\mid i}\perp v_i\). Thus, if a vertex $i$ is removed by the matching with constant probability conditional on $i\in S$, then there is a family of edge directions $\{u_{j\mid i}:j\in N_G(i)\}$ whose maximum projection on $\gamma$ exceeds \(\sqrt3\,t\) with constant probability. This observation motivates the following definition.

\begin{definition}[$(s,\delta)$-cover]
A set of vectors $X = \{x_1,\ldots,x_k\} \subseteq \R^d$ is called an $(s,\delta)$-cover if for
$\gamma \sim \mathcal N({0, \Id_d})$, we have
\(\Pr\left[\exists i : \langle \gamma, x_i\rangle \ge s\right] \ge \delta.\)
\end{definition}

The analysis of \cite{ACC06} then proceeds by showing that the failure of \kmsp rounding forms a subgraph of ``bad'' vertices where the neighborhoods form covers. More specifically, we have the following lemma.

\begin{lemma}[Lemma 3 and 4 of~\cite{ACC06}]
\label{lem:kms-core}
Let $I$ be the output of \kmsp algorithm with input graph $G = (V,E)$, strict vector
\(3\)-coloring \(\{v_i\}_{i\in V}\) and threshold $t$.
At least one of the
following holds:
\begin{enumerate}[(i)]
\item $\E|I|\ge\frac{n}{4} \Phi^c(t)$
\item There is a nonempty induced subgraph \(G'=(V',E')\) of $G$ such that, for
every \(i\in V'\), $\{u_{j\mid i}\}_{j\in N_{G'}(i)}$ is a $(\sqrt{3}t,\frac{1}{8})$-cover.
\end{enumerate}
\end{lemma}
Suppose condition (i) does not hold, which we refer to as \kmsp fails. The idea of \cite{ACC06} is to propagate the cover structure to the next level of neighborhoods.
Fix a root vertex $i$, and consider a $2$-level walk $i\to j\to k$ in $G'$. We define its type and orthogonal component by
\(y_{jk}:=\ip{u_{i\mid j}}{u_{k\mid j}}\) and
\(u_{k\mid i,j}:=u_{k\mid j}-y_{jk}u_{i\mid j}\), respectively.
Since \(u_{i\mid j}\) and \(u_{k\mid j}\) are unit vectors orthogonal to \(v_j\), we have
\(u_{k\mid i,j}\perp\operatorname{span}\{v_i,v_j\}\). By direct substitution,
\begin{align*}
&v_k
=\left(\frac14+\frac34y_{jk}\right)v_i
-\frac{\sqrt3}{4}(1-y_{jk})u_{j\mid i}
+\frac{\sqrt3}{2}u_{k\mid i,j},\\
&v_k-\ip{v_i}{v_k}v_i
=\frac{\sqrt3}{4}(1-y_{jk})(-u_{j\mid i})
+\frac{\sqrt3}{2}u_{k\mid i,j}.
\end{align*}

Lemma~\ref{lem:kms-core} ensures that $\{u_{j\mid i}\}_{j\in N_{G'}(i)}$ and $\{u_{k\mid j}\}_{k\in N_{G'}(j)}$ are both good covers. Thus, $\{-u_{j\mid i}\}_{j\in N_{G'}(i)}$ and $\{u_{k\mid i,j}\}_{k\in N_{G'}(j)}$ are good covers by symmetry and the fact that projecting away one direction does not significantly affect the cover strength (Lemma \ref{lem:projection}). Formally,~\cite{ACC06} shows the following lemma to combine two good covers.

\begin{lemma}[Cover composition~\cite{ACC06,chlamtac2009non}]
\label{lem:acc-cover-composition}

Fix \(c>0\) and let \(s_1\to\infty\). Let
\(X\subseteq\R^d\)
be an \((s_1,\delta_1)\)-cover consisting of unit vectors with size
\(|X| \le \bigl(\Phi^c(s_1)\bigr)^{-(1+c)}.\)
Suppose
\(Y_x\subseteq x^\perp\)
is an \((s_2,\delta_2)\)-cover for every \(x\in X\), where \(0<\delta_2\le1/2\).
Then,
\begin{align*}
\Pr_{\gamma\sim\mathcal N(0,\Id)}
\left[
    \begin{gathered}
    \exists\,x\in X,\ y\in Y_x\ \text{such that}\\
    \langle\gamma,x\rangle\ge s_1,
    \langle\gamma,y\rangle
    \ge
    s_2
    -
    \|y\|
    \left(
        q_2+
       \sqrt{c(1+\varepsilon)}s_1
    \right)
    \end{gathered}
\right]
&\ge \delta_1-O(\tfrac{1}{s_1}),
\end{align*}
where
\(\Phi^c(q_2)=\delta_2\) and $\varepsilon = O(\frac{\log s_1}{s_1^2})$.
\end{lemma}
By applying this lemma, and with some additional effort, one can show that there exists a root $i$, a subset of vertices $K \subseteq N_{G'}^{(2)}(i)$, and a constant $a>0$ such that the set
\(\left\{\frac{v_k-\ip{v_i}{v_k}v_i}{\sqrt{1-\ip{v_i}{v_k}^2}}\right\}_{k\in K}\) is a $(a\sqrt{3}t,\Omega(1))$-cover. If $a\sqrt3\,t$ exceeds $\sqrt{2\log n}$ by a fixed positive multiple of $t$, this is a contradiction, as a direct union bound gives
\begin{align*}
\Omega(1)
&\le\Pr_\gamma\left[
\exists k\in K:
\ip{\gamma}{\frac{v_k-\ip{v_i}{v_k}v_i}{\sqrt{1-\ip{v_i}{v_k}^{\,2}}}}
\ge a\sqrt3\,t-o(t)
\right]\\
&\le n\Phi^c\bigl(a\sqrt3\,t-o(t)\bigr)=o(1).
\end{align*}
It therefore follows that \kmsp returns an independent set of the required size in expectation.

Bansal, Huang, and Lee~\cite{BHL26} extend this analysis to third-level neighborhoods, obtaining large independent sets from suitable subsets of second- or third-level neighborhoods. Their third-level argument combines cover composition with an improved vector coloring obtained from a stronger sum-of-squares relaxation, followed by KMS rounding.

Our analysis extends the neighborhood argument to every fixed number of levels, but uses the resulting structure differently. Rather than rounding an improved vector coloring of a higher-level neighborhood, we propagate the cover structure and use a potential function to show that these cover structures cannot persist. This establishes a larger expected output for KMS$'$ itself. We outline this argument next.

\subsection{Outline of our Proof} \label{subsec:outline_proof}

Our algorithm is based on a refined analysis of the KMS$'$ Gaussian rounding procedure. Our contribution is the following theorem in the sparse regime.

\begin{theorem}
\label{thm:sparse-progress}
Fix \(0<c<c^\dagger := \frac{\sqrt{97}-7}{12}\). For all sufficiently large \(n\), there is a
randomized polynomial-time algorithm that, given a \(3\)-colorable
\(n\)-vertex graph \(G\) satisfying
\(\Delta(G)\le n^{3(1+c)/(5+3c)}\), with high probability returns an independent set of size $\Omega\left(n^{1-\frac{1}{5+3c}}\right)$.
\end{theorem}

Our improvement comes from running \kmsp with a lower threshold, thereby producing a larger initial random set, and analyzing the failure of this more aggressive rounding recursively. In particular, we choose $t$ so that $\Phi^c(t) \approx n^{-1/(5 + 3c)}$.
If sufficiently many selected vertices survive the maximal-matching cleanup, we immediately obtain the desired independent set. Otherwise, as in Lemma \ref{lem:kms-core}, the matching yields a cover at every vertex of a nonempty induced subgraph.

We show that each $1$-level certificate can be propagated to a cover at an arbitrary fixed number of neighborhood levels. We then construct a potential function that increases by a uniform positive amount at every step. Since all feasible certificates lie in a fixed compact region, this increase cannot continue indefinitely. Thus, the failure certificate cannot exist, forcing the more aggressive rounding to return a large independent set. Making this recursive argument work requires several new technical ideas, which we describe next.

Our starting point is to identify the information that must be retained for each certificate. Fix a root \(i\), and consider a
certificate supported on the endpoints \(j\) of walks beginning at \(i\). To
extend this certificate by one level, we append an edge \(j\to k\) and seek a
new certificate supported on the resulting endpoints \(k\). The geometry of
this extension depends on the correlation between \(v_j\) and the fixed root
vector \(v_i\), and the component of \(v_j\) orthogonal to \(v_i\). These quantities necessarily
interact: changing the correlation with the root changes the length of the
orthogonal component and hence the normalization of the Gaussian threshold.
This motivates the recursive Gaussian certificate formalized in
Definition~\ref{def:recursive-gaussian-certificate}. Such a certificate rooted
at \(i\) records a common root correlation \(a\), a Gaussian threshold \(s\),
and a cover probability \(\delta\). A level-\(\ell\) certificate is supported
on endpoints of length-\(\ell\) walks from \(i\). In particular,
Lemma~\ref{lem:kms-core} gives the level-\(1\) certificate
\(S^{(1)}:=N_{G'}(i)\) with parameters
\(\left(-\frac12,\sqrt3\,t,\frac18\right)\).

To propagate a certificate from level \(\ell\) to level \(\ell+1\), we
extend each walk by one edge \(j\to k\). These extensions need not
initially have the same correlation with the root. However, once the
correlation \(a\) of \(v_j\) with the root \(v_i\) is fixed, the correlation
of \(v_k\) with \(v_i\) is determined by a single scalar \(y\), which we will refer to as the type of the walk, measuring
how strongly the new edge direction points back toward the root. With some computation, the new correlation is
\(a'=-\frac a2+\frac{\sqrt3}{2}\sqrt{1-a^2}\,y\). We partition the possible values of \(y\) into short intervals and retain
one interval on which the cover guarantees remain valid. Since
there are only subpolynomially many intervals, this restriction incurs only
a negligible loss. Then, up to \(o(1)\) error, all retained extended walks have the same type \(y\), and hence the endpoints of these walks have the same new correlation \(a'\). To show that the cover structure also persists under this propagation, we prove a more refined cover-composition
lemma than the one in \cite{ACC06}.

 Having established that recursive Gaussian certificates can be propagated,
we use this recursion to rule out failure of KMS$'$. Every
recursive certificate consists of at most \(n\) unit vectors. Therefore, a
union bound gives
\begin{align*}
\Pr\left[
\max_{j\in S^{(\ell)}}\langle\gamma,x_j\rangle
\ge\lambda t-o(t)
\right]
&\le n\Phi^c\bigl(\lambda t-o(t)\bigr).
\end{align*}
Since $\log n=\frac{5+3c}{2}t^2+o(t^2),$ the right-hand side is \(\exp(-\Omega(t^2))\) whenever $\lambda>\sqrt{5+3c}+\Omega(1),$ which will contradict the \(\exp(-o(t^2))\) lower bound we attain from recursive Gaussian certificates.

It would thus suffice to show that \(\lambda\) increases by a fixed amount
at every propagation step. However, this is not generally true. The
threshold is measured after removing the component in the root direction
and normalizing the remainder. Since this normalization changes with the
correlation \(a\), progress contributed by the new edge can be offset by a
change in correlation. Thus, \(\lambda\) alone does not capture the progress
made by the recursion.

We account for this interaction using the potential function
\begin{align*}
V(a,\lambda)
&:=\lambda-\frac{3(1+c)}{1+4c}
\left(\frac{1}{\sqrt{1-a^2}}-1\right).
\end{align*}
The correction term compensates for changes in the normalized threshold
caused by changes in the correlation with the root. In
Section~\ref{subsec:properties}, we explain why this potential is
natural and derive its precise form by examining the recursive transition. We then prove that every propagation step satisfies
\(V(a',\lambda')-V(a,\lambda)\ge\frac{\mu(c)}2\) for some \(\mu(c)>0\) whenever \(c<c^\dagger\).

Under failure of KMS$'$, the recursion could be continued through any fixed
number of levels, while the potential would increase by a uniform positive
amount at every step. This is impossible because \(V\) is bounded on the
compact region of feasible parameters. We conclude that the failure
certificate cannot exist, forcing KMS$'$ to return an independent set of
the size asserted in Theorem~\ref{thm:sparse-progress}.

\subsection{Organization of the Paper} \label{subsec:organization}

In Section~\ref{sec:prelim}, we first show that cover structures can be preserved under projection (Lemma~\ref{lem:projection}). Then, we prove a more refined cover-composition lemma than the one in~\cite{ACC06} (Lemma~\ref{lem:orthogonal-cover-composition}). In Section~\ref{subsec:failure-certificate}, we first introduce the notion of recursive Gaussian certificates (Definition~\ref
{def:recursive-gaussian-certificate}) and then show how to propagate them from one level to the next (Lemma~\ref{lem:level-2-certificate}, Lemma~\ref{lem:recursive-step}). In Section~\ref{sec:potential-function}, we define the potential function with motivation in Section~\ref{subsec:properties} and give the proof of Theorem~\ref{thm:sparse-progress} in Section~\ref{subsec:proof-of-main}. Appendix~\ref{subsec:dense-sparse-balance}
reviews Blum's notion of progress and applies the sparse--dense combination
framework to deduce Theorem~\ref{thm:main} from
Theorem~\ref{thm:sparse-progress}.

\section{Recursive Gaussian Certificates} \label{sec:main}

\subsection{Composition of Covers} \label{sec:prelim}

The main point of this section is to show that covers can be preserved under two operations: projection and composition. The projection lemma is standard and appears, for example, as Lemma~3.4.2 in Chlamtac's thesis~\cite{chlamtac2009non}.

\begin{lemma}[Lemma 3.4.2~\cite{chlamtac2009non}]\label{lem:projection}
Let \(X\) be an \((s,\delta)\)-cover in $\R^d$ consisting of vectors of norm
at most one and \(z\in\R^d\) be a unit vector.  For every \(\varepsilon \ge0\), the
projected family \(X^\perp := \left\{ x-\langle x,z\rangle z:x\in X \right\}\) is an $(s- \varepsilon,\delta-2\Phi^c(\varepsilon))$-cover.
\end{lemma}

\begin{proof} Let $\gamma \sim \mathcal{N}(0, \Id_d)$. Decompose \(\gamma=\gamma_\perp+\xi z\), where $\langle\gamma_\perp, z\rangle = 0$. We know that $\xi = \langle\gamma,z\rangle\sim \mathcal N(0,1)$ and $\gamma_\perp$ is independent of $\xi$. Define $E_1$ to be the event that $\max_{x\in X}\langle\gamma,x\rangle\ge s$, and $E_2$ to be the event that $|\xi|\le \varepsilon$. By definition of $X$ being an $(s,\delta)$-cover, we have
\begin{align*}
\mathbb{P}\left[\max_{x\in X}\left\langle\gamma_\perp,x-\langle x,z\rangle z\right\rangle\ge s-\varepsilon\right]
\ge \mathbb{P}[E_1\cap E_2] 
\ge \mathbb{P}[E_1]-\mathbb{P}[E_2^c]
\ge \delta-2\Phi^c(\varepsilon).
\end{align*}
\end{proof}

The cover-composition lemma in this section is a refined version of the corresponding lemma in~\cite{ACC06} (Lemma~\ref{lem:acc-cover-composition}). 
Before presenting our lemma, we first introduce some necessary background. The new ingredient is Ehrhard's inequality, which is stated as follows.

\begin{fact}[Ehrhard's Inequality \cite{Ehrhard83}]
    Let \(\gamma_d\) denote standard Gaussian measure on \(\mathbb R^d\).
    Given convex Borel sets \(A,B\subseteq\mathbb R^d\) with
    \(0<\gamma_d(A),\gamma_d(B)<1\), and \(\theta\in[0,1]\), we have
\begin{align}
\Phi^{-1}\!\left(\gamma_d\bigl(\theta A+(1-\theta)B\bigr)\right)
&\ge \theta\Phi^{-1}\!\left(\gamma_d(A)\right)
+(1-\theta)\Phi^{-1}\!\left(\gamma_d(B)\right).
\label{eq:ehrhard}
\end{align}
\end{fact}

We also state the following fact on bounding the tail probability of a standard Gaussian, which will be useful in our analysis.
\begin{fact}[Gaussian Tail Bounds and Mills' Ratio]\label{fact:tail-bound}
    For every \(t>0\),
\begin{align}
    \frac{t}{1+t^2}\,\frac{1}{\sqrt{2\pi}}\exp\left(-\frac{t^2}{2}\right)
    \;\le\;
    \Phi^c(t)
    \;\le\;\exp\!\left(-\frac{t^2}{2}\right).
    \label{eq:mills}
\end{align}
\end{fact}

We first present a lemma on upper bounding the lower tail of a Gaussian maximum. Given a vector $\gamma \in \mathbb{R}^d$ and a subset of vectors $Y=\{y_1,\ldots,y_N\}\subseteq\mathbb R^d$, define 
\begin{align}
M_Y(\gamma)&:=\max_{r\le N}\langle\gamma,y_r\rangle.
\label{eq:M_definition}
\end{align}

The following lemma is a direct consequence of the Ehrhard concavity of Gaussian
maxima together with  standard Gaussian tail estimates.

\begin{lemma}[Lower tail of a Gaussian maximum]
\label{lem:gaussian-maximum-lower-tail}
Let $\gamma \sim \mathcal{N}(0, \Id_d)$ and
\(
Y=\{y_1,\ldots,y_N\}\subseteq\mathbb{R}^d
\)
be a finite collection of vectors satisfying
$\|y_r\|_2\le 1$ for every $r\in[N]$.
Let $M_Y(\gamma)$ be defined as in~\eqref{eq:M_definition}.
Fix $t > 0$, $T>t$, and $L > 0$. Define
\(
p:=\Pr[M_Y(\gamma)>t]\) and \(
\delta:=N\Phi^c(T).
\)
If $\delta\in(0,1)$, then
\begin{align*}
\Pr[M_Y(\gamma)\le t-L]
&\le \Phi\left(\left(1+\frac{L}{T-t}\right)\Phi^{-1}(1-p)
-\frac{L}{T-t}\Phi^{-1}(1-\delta)\right).
\end{align*}
\end{lemma}

\begin{proof}
If all $y_r=0$, then $p=0$ and the right-hand side is $1$,
so the claim is immediate. Assume henceforth that some $y_r\ne0$.
For $u\in\mathbb R$, define
\begin{align*}
K_u:=\{\gamma\in\mathbb R^d:M_Y(\gamma)\le u\},\quad
 F(u):=\Phi^{-1}\!\left(\Pr[M_Y(\gamma)\le u]\right).
\end{align*}
Each $K_u$ is convex. And $\theta K_u+(1-\theta)K_v\subseteq K_{\theta u+(1-\theta)v}$ for $u,v\in\mathbb R$ and $0<\theta<1$. Ehrhard's inequality therefore implies that $F$ is concave. Since some $y_r\ne0$ and $t,T>0$, both $F(t)$ and $F(T)$ are finite. Evaluating at $t$, we have
\begin{align*}
F(t)&=\Phi^{-1}\!\left(\Pr[M_Y(\gamma)\le t]\right)
=\Phi^{-1}(1-p).
\end{align*}
Evaluating at $T$, we have
\begin{align*}
F(T)&=\Phi^{-1}\!\left(1-\Pr[M_Y(\gamma)>T]\right) \\
&\ge \Phi^{-1}\!\left(1-\sum_{r=1}^N\Pr[\langle\gamma,y_r\rangle>T]\right) \\
&\ge \Phi^{-1}\!\left(1-N\Phi^c(T)\right)\\
&=\Phi^{-1}(1-\delta),
\end{align*}
where the first inequality is a union bound and the second uses the fact that $\|y_r\|\le 1$.
Since $t-L \leq t \leq T$, concavity of $F$ gives
\begin{align*}
F(t)
&\ge
\frac{T-t}{T-t+L}F(t-L)
+
\frac{L}{T-t+L}F(T).
\end{align*}
Then, we have
\begin{align*}
F(t-L)
&\le
\left(1+\frac{L}{T-t}\right)F(t)
-\frac{L}{T-t}F(T)\\
&\le
\left(1+\frac{L}{T-t}\right)\Phi^{-1}(1-p)
-\frac{L}{T-t}\Phi^{-1}(1-\delta).
\end{align*}
Applying $\Phi$ to both sides of the inequality gives the result.
\end{proof}

We are now ready to present the cover-composition lemma. 

\begin{lemma}[Cover composition]
\label{lem:orthogonal-cover-composition}
Fix \(c>0\) and 
let \(t\to\infty\). Let $\lambda=\lambda(t),E=E(t), $ and $U=U(t)$ be uniformly bounded, with \[0\le\lambda\le\sqrt E  \qquad \text{ and } \qquad U^2\ge 1/(1+c).\]
Let \(X\subseteq\R^d\) be a finite \((\lambda t-o(t),\exp(-o(t^2)))\)-cover satisfying \[ \sup_{x\in X}\|x\|\le 1
, \quad \log|X|\le \frac{E}{2}t^2+o(t^2).
\] 
For each \(x\in X\), let \(Y_x\subseteq x^\perp\) be a finite \((\sqrt3\,t-o(t),\exp(-o(t^2)))\)-cover satisfying \[\sup_{y\in Y_x}\|y\|\le U+o(1),\quad \log|Y_x|\le \frac{3(1+c)}{2}t^2+o(t^2),\]
uniformly over \(x\in X\). 
Then, 
\[\Pr_{\gamma\sim\mathcal N(0,\Id_d)}\!\left[
\exists x\in X,\ \exists y\in Y_x:
\begin{array}{l}
\langle\gamma,x\rangle\ge\lambda t-o(t),\\
\langle\gamma,y\rangle\ge
\left(\sqrt3-\sqrt{(E-\lambda^2)
\left(U^2-\frac1{1+c}\right)}\right)t-o(t)
\end{array}
\right]\ge \exp(-o(t^2)).\]
\end{lemma}

\begin{proof}
Choose a sequence \(\{\xi_t\}\downarrow0\) such that all $o(1)$ terms above have absolute value at most $\xi_t$, all $o(t)$ terms above have absolute value at most $\xi_t t$, and all $o(t^2)$ terms have absolute value at most $\xi_t t^2$, simultaneously for all $x\in X$. Choose \(\{\sigma_t\}\downarrow0\) such that \(\xi_t=o(\sigma_t^3)\) and \(\sigma_t^2t^2/\log t\to\infty\). Define
\begin{align*}
\bar{E}_t:=E+\sigma_t,\quad \bar{U}_t:=\sqrt{U^2+\sigma_t},\quad  L_t:=(1+\sigma_t)\sqrt{(\bar{E}_t-\lambda^2)
\left(\bar{U}_t^2-\frac1{1+c}\right)}\,t.
\end{align*}
From boundedness of $\lambda,E$, and $U$, we know that \(L_t = \sqrt{(\bar{E}_t-\lambda^2)
\left(\bar{U}_t^2-\frac1{1+c}\right)}\,t + o(t)\). For sufficiently large $t$, we have \(\sup_{y\in Y_x}\|y\|\le\bar{U}_t\) for all $x\in X$.
For large enough $t$, \(x\in X\), and a random vector $\gamma\in \R^d$, define the following two events:
\[A_{x,t}:=\{\langle\gamma,x\rangle\ge(\lambda-\xi_t)t\},\qquad B_{x,t}:=\{M_{Y_x}(\gamma)<(\sqrt3-\xi_t)t-L_t\}.\] 
We have
\begin{equation}\label{eq:asymptotic-composition-split}
\begin{aligned}
    \Pr_{\gamma\sim\mathcal N(0,\Id_d)}&\!\left[
\exists x\in X,\ \exists y\in Y_x:
\begin{array}{l}
\langle\gamma,x\rangle\ge\lambda t-\xi_t t,\\
\langle\gamma,y\rangle\ge (\sqrt3-\xi_t)t-L_t
\end{array}
\right] \\&=\Pr\left[\bigcup_{x\in X}(A_{x,t}\cap B_{x,t}^c)\right]
\\&\ge\Pr\left[\bigcup_{x\in X}A_{x,t}\right]
-\Pr\left[\bigcup_{x\in X}(A_{x,t}\cap B_{x,t})\right]\\
&\ge\exp(-\xi_t t^2)
-\Pr\left[\bigcup_{x\in X}(A_{x,t}\cap B_{x,t})\right],
\end{aligned}
\end{equation}
where the last inequality follows because $X$ is a \((\lambda t-o(t),\exp(-o(t^2)))\)-cover. It remains to upper-bound \(\Pr\left[\bigcup_{x\in X}(A_{x,t}\cap B_{x,t})\right]\). For every \(x\in X\), define the rescaled set
\(\widetilde Y_{x,t}:=\{y/\bar{U}_t:y\in Y_x\}\). Then \(\widetilde Y_{x,t}\) consists of vectors of norm at most one.
We write
\[
 s_t:=\frac{(\sqrt3-\xi_t)t}{\bar{U}_t},\quad r_t:=\frac{\bar{U}_t^2(3(1+c)+2\xi_t)}{(\sqrt3-\xi_t)^2}.
\]
We know that \(r_t>1\), and \(|Y_x|\le\exp(r_ts_t^2/2)\).
Let \(p_{x,t}:=\Pr[M_{\widetilde Y_{x,t}}(\gamma)>s_t]\) and \(\delta_{x,t}:=|Y_x|\Phi^c(r_ts_t)\). We first upper bound \(\Phi^{-1}(1-p_{x,t})\) as follows:
 \begin{equation}\label{eq:asymptotic-px}
     \Phi^{-1}(1-p_{x,t})\leq \Phi^{-1}(1 - \exp(-\xi_t t^2)) \leq \sqrt{2\xi_t} t,
 \end{equation}
 where the first inequality follows because $Y_x$ is a \((\sqrt3\,t-o(t),\exp(-o(t^2)))\)-cover and the second follows from the upper bound in Fact~\ref{fact:tail-bound}.

 Then we lower bound \(\Phi^{-1}(1-\delta_{x,t})\). Let $\theta_t = \frac{\sigma_t}{4(1+\sigma_t)}$. Since $\bar{U}_t^2 - \frac{1}{1+c}\geq \sigma_t$ and  \(\sigma_t^2t^2/\log t\to\infty\), one can check that
 $\frac{\theta_t (r_t^2s_t^2-r_ts_t^2)}{\log t} \to \infty$ as $t\to \infty$. Then for large enough $t$,
 we have
 \begin{equation}\label{eq:asymptotic-deltax}
     \begin{aligned}
 \delta_{x,t} \leq \exp\left(\frac{r_ts_t^2-r_t^2s_t^2}{2}\right) \leq \Phi^c\left((1-\theta_t)\sqrt{r_t(r_t-1)} s_t\right),
     \end{aligned}
 \end{equation}
 where the first inequality follows from $|Y_x|\leq \exp\left(\frac{3(1+c)}{2} t^2 + \xi_t t^2\right)$ and the upper bound in Fact~\ref{fact:tail-bound}, and the second follows from the lower bound in Fact~\ref{fact:tail-bound}. Thus,
 \[
 \Phi^{-1}(1 - \delta_{x,t}) \geq (1-\theta_t)\sqrt{r_t(r_t-1)} s_t.
 \]
 We now apply Lemma~\ref{lem:gaussian-maximum-lower-tail}.
\begin{equation}
    \begin{aligned}
        \Pr[B_{x,t}] &= \Pr[M_{\widetilde Y_{x,t}}(\gamma)< s_t - \frac{L_t}{\bar{U}_t}]\\
        &\leq \Phi\left(\left(1+\frac{L_t}{\bar{U}_t (r_t-1) s_t }\right)\Phi^{-1}(1 - p_{x,t}) - \frac{L_t}{\bar{U}_t (r_t-1) s_t } \Phi^{-1}(1 - \delta_{x,t})\right) \\
        &\leq \Phi\left(\left(1+\frac{L_t}{\bar{U}_t (r_t-1) s_t }\right)\sqrt{2\xi_t} t - (1-\theta_t)\frac{L_t}{\bar{U}_t } \sqrt{\frac{r_t}{r_t-1}}\right).
    \end{aligned}
\end{equation}

We now bound the two terms in the argument of $\Phi$. First, observe that
\begin{align*}
r_t-1
&=
\frac{
\bar U_t^2(3(1+c)+2\xi_t)-(\sqrt3-\xi_t)^2
}{
(\sqrt3-\xi_t)^2
}\\
&=
\frac{
3(1+c)\left(\bar U_t^2-\frac1{1+c}\right)
+2\xi_t\bar U_t^2+2\sqrt3\,\xi_t-\xi_t^2
}{
(\sqrt3-\xi_t)^2
}\\
&= (1+c)\left(\bar U_t^2-\frac1{1+c}\right) + O(\xi_t)\\
&= (1+c)\left(\bar U_t^2-\frac1{1+c}\right)(1+o(1)).
\end{align*}
Therefore,
\begin{align}
\frac{L_t}{\bar U_t(r_t-1)s_t}
&=
O\left(
\sqrt{
\frac{\bar E_t-\lambda^2}
{\bar U_t^2-\frac1{1+c}}
}
\right).
\label{eq:Lt-ratio-order}
\end{align}

It follows from \eqref{eq:Lt-ratio-order} and the choice of $\bar{E}_t$ and $\bar{U}_t$ that
\begin{align*}
&\frac{
\left(
1+\frac{L_t}{\bar U_t(r_t-1)s_t}
\right)\sqrt{2\xi_t}\,t
}{
\sigma_t\sqrt{\bar E_t-\lambda^2}\,t
}\le
C\left(
\frac{\sqrt{\xi_t}}
{\sigma_t\sqrt{\bar E_t-\lambda^2}}
+
\frac{\sqrt{\xi_t}}
{\sigma_t\sqrt{\bar U_t^2-\frac1{1+c}}}
\right) \leq 2C\frac{\sqrt{\xi_t}}{\sigma_t^{3/2}}
\end{align*}
for some constant $C>0$ independent of $t$. 
The assumption $\xi_t=o(\sigma_t^3)$ implies that
\begin{equation}\label{eq:first-inner-term}
\left(
1+\frac{L_t}{\bar U_t(r_t-1)s_t}
\right)\sqrt{2\xi_t}\,t
=
o\left(
\sigma_t\sqrt{\bar E_t-\lambda^2}\,t
\right).
\end{equation}
We next bound the second term. By the definition of $r_t$,
\begin{align*}
\bar U_t^2\frac{r_t-1}{r_t}
=
\bar U_t^2
-
\frac{(\sqrt3-\xi_t)^2}
{3(1+c)+2\xi_t}
=
\bar U_t^2-\frac1{1+c}+O(\xi_t) =
\left(\bar U_t^2-\frac1{1+c}\right)
(1+o(\sigma_t)).
\end{align*}
Therefore,
\begin{align*}
(1-\theta_t)\frac{L_t}{\bar U_t}
\sqrt{\frac{r_t}{r_t-1}}
&=
(1-\theta_t)
(1+\sigma_t)
\sqrt{
(\bar E_t-\lambda^2)
\left(\bar U_t^2-\frac1{1+c}\right)
}\,t
\frac{1+o(\sigma_t)}
{\sqrt{\bar U_t^2-\frac1{1+c}}}\\
&=
(1-\theta_t)(1+\sigma_t)
(1+o(\sigma_t))
\sqrt{\bar E_t-\lambda^2}\,t\\
&= \left(
1+\frac{3\sigma_t}{4}+o(\sigma_t)
\right)
\sqrt{\bar E_t-\lambda^2}\,t.
\end{align*}

Combining \eqref{eq:first-inner-term} with the preceding estimate for the
second term, we obtain, for all sufficiently large \(t\),
\begin{align}
\Pr[B_{x,t}]
&\le\Phi\left(-\left(1+\frac{\sigma_t}2\right)
\sqrt{\bar{E}_t-\lambda^2}\,t\right)\notag\\
&\le\exp\left(-\frac{(1+\sigma_t/2)^2(\bar{E}_t-\lambda^2)}2t^2\right).
\label{eq:asymptotic-bx}
\end{align}

Note that all the preceding estimates hold uniformly over \(x\in X\). Since \(Y_x\subseteq x^\perp\), the events \(A_{x,t}\) and \(B_{x,t}\) are independent. Combining \eqref{eq:asymptotic-bx} with a union bound and \(\log|X|\le\bar{E}_t t^2/2+\xi_t t^2\), we obtain
\begin{align*}
\Pr\left[\bigcup_{x\in X}(A_{x,t}\cap B_{x,t})\right]
&\le\sum_{x\in X}\Pr[A_{x,t}]\Pr[B_{x,t}]\\
&\le |X| \max_{x\in X} \Pr[A_{x,t}] \max_{x\in X} \Pr[B_{x,t}]
\\
&\le\exp\left(\left[
\frac{\bar{E}_t}{2}+\xi_t-\frac{(\lambda-\xi_t)_+^2}{2}
-\frac{(1+\sigma_t/2)^2(\bar{E}_t-\lambda^2)}{2}
\right]t^2\right)\\
&\le\exp\left(-\frac{\sigma_t^2}{3}t^2\right).
\end{align*}
The last inequality follows from \(\bar{E}_t-\lambda^2\ge\sigma_t\), boundedness of \(\lambda\), and \(\xi_t=o(\sigma_t^3)\). Substituting into \eqref{eq:asymptotic-composition-split} gives
\begin{align*}
\Pr\left[\bigcup_{x\in X}(A_{x,t}\cap B_{x,t}^c)\right]
&\ge\exp(-\xi_t t^2)-\exp\left(-\frac{\sigma_t^2}{3}t^2\right)
=\exp(-o(t^2)).
\end{align*}

By definition of $L_t$, $\xi_t$, the thresholds satisfy the ones in the statement of the lemma. This completes the proof.

\end{proof}

We briefly comment on the differences between our approach and the one in~\cite{ACC06}. Our lemma takes the cardinality of the inner cover $Y_x$ into account. This enables us to obtain the high threshold $r_ts_t$ to apply Lemma~\ref{lem:gaussian-maximum-lower-tail}, which is a more refined estimate than the one in~\cite{ACC06}. 

\subsection{Certificates from Longer Walks}
\label{subsec:failure-certificate}

For the remainder of the analysis, suppose condition (i) in Lemma~\ref{lem:kms-core} fails and thus alternative (ii) holds. Fix the resulting subgraph \(G'\).
Throughout this subsection, fix \(0<c<c^\dagger\), assume
\(\Delta(G)\le n^{3(1+c)/(5+3c)}\), and choose \(t\) so that \(\Phi^c(t)=n^{-1/(5+3c)}\). Our main new ingredient is the recursive reuse of the resulting local covers. We begin with the key definition that enables us to propagate the failure certificate.

\begin{definition}[Recursive Gaussian certificate]
\label{def:recursive-gaussian-certificate}
Let \(G=(V,E)\) be a $3$-colorable graph and \(\{v_j\}_{j\in V}\) its strict vector \(3\)-coloring.
Fix a root vertex \(i\in V\). Given \(a\in(-1,1)\), a set
\(S\subseteq V\) is an \emph{\((a,s,\delta)\)-recursive Gaussian certificate
rooted at \(i\)} if:
\begin{enumerate}
    \item  $\langle v_i,v_j\rangle=a+o(1)$ uniformly over $j\in S,$
    \item The normalized directions $\left\{
        \frac{v_j-\langle v_i,v_j\rangle v_i}{\sqrt{1-\langle v_i,v_j\rangle^2}}: j\in S \right\}$ form an \((s,\delta)\)-cover.
\end{enumerate}
\end{definition}

If every \(j\in S\) is the endpoint of a length-\(\ell\) walk in \(G'\) beginning at \(i\), then the certificate is said to be a level-\(\ell\) certificate.  
Since we suppose that condition (i) in Lemma~\ref{lem:kms-core} fails, condition (ii) gives a level-\(1\) certificate. Indeed, for a fixed root \(i\in V'\), the set $S^{(1)}:=N_{G'}(i)$ forms a \((-\frac{1}{2}, \sqrt{3} t, \frac{1}{8})\)-recursive Gaussian certificate. The proof idea is to extend the certificate to higher levels to obtain a contradiction.  

We now describe the construction of the level-\(2\) certificate. Starting from \(j\in S^{(1)}\), we first extend the edge \(i\to j\) by choosing a neighbor
\(k\in N_{G'}(j)\) to produce a length-\(2\) walk $i\rightarrow j\rightarrow k.$
Since both \(i\) and \(k\) are adjacent to \(j\), we have
\begin{align*}
v_i=-\frac12v_j+\frac{\sqrt3}{2}u_{i\mid j}, \quad v_k=-\frac12v_j+\frac{\sqrt3}{2}u_{k\mid j}.
\end{align*}
We define the type of the walk \(i\to j\to k\) by
\begin{align}
y_{jk}&:=\langle u_{i\mid j},u_{k\mid j}\rangle.
\label{eq:level-two-type}
\end{align}
The type also characterizes the correlation of \(v_k\) with the root. Indeed,
\begin{align}
\langle v_i,v_k\rangle
&=\left\langle-\frac12v_j+\frac{\sqrt3}{2}u_{i\mid j},
-\frac12v_j+\frac{\sqrt3}{2}u_{k\mid j}\right\rangle\notag\\
&=\frac14+\frac34\langle u_{i\mid j},u_{k\mid j}\rangle
=\frac14+\frac34y_{jk}.
\label{eq:level-two-correlation}
\end{align}

The next lemma shows that we can construct a good level-\(2\) certificate from \(S^{(1)}\) by pruning the set of length-\(2\) walks to those whose types lie in a short interval \(J\subseteq[-1,1]\).

\begin{lemma}[The level-\(2\) certificate]
\label{lem:level-2-certificate}
Let \(t\) be the threshold in
Algorithm~\ref{alg:kms-prime} and set \(\rho_t:=(\log t)^{-2}\). 
For every fixed \(\tau>0\), there exist an interval
\(J\subseteq[-1,1]\) of width at most \(\rho_t\), a subcollection
\(\widetilde S^{(1)}\subseteq S^{(1)}\), and a number
\(y\in[-\sqrt{c/(1+c)},\sqrt{c/(1+c)}]\) such that
\begin{align}
S^{(2)}&:=\left\{k\in V':\exists j\in\widetilde S^{(1)},\quad
k\in N_{G'}(j),\quad y_{jk}\in J\right\}
\label{eq:level-two-endpoints}
\end{align}
is a \(\left( b_2(y), \bigl(\lambda_2(y)-\tau\bigr)t-o(t), \exp\bigl(-o(t^2)\bigr) \right)\)-recursive Gaussian certificate rooted at \(i\), where 
\begin{align*}
b_2(y)&:=\frac14+\frac34y, \\
\lambda_2(y)&:=\frac{\frac94-\frac34y
-\frac{\sqrt3}{2}\sqrt{(2+3c)\left(\frac{c}{1+c}-y^2\right)}}
{\sqrt{1-b_2(y)^2}}.
\end{align*}
\end{lemma}

\begin{proof}
Partition \([-1,1]\) into consecutive intervals \(\mathcal J_t\), each of
width at most \(\rho_t\). Therefore, $|\mathcal J_t|\le\frac3{\rho_t}.$ For \(j\in S^{(1)}\) and \(I\in\mathcal J_t\), let \(N_j(I) := \left\{ k\in N_{G'}(j):y_{jk}\in I \right\}\). The condition (ii) in Lemma~\ref{lem:kms-core} gives that
\begin{align*}
\Pr_\gamma\!\left[\max_{k\in N_{G'}(j)}\langle\gamma,u_{k\mid j}\rangle\ge\sqrt3\,t\right]
&\ge\frac18.
\end{align*}
Since $N_{G'}(j)
=
\bigcup_{I\in\mathcal J_t}N_j(I),$ the union bound and the upper bound \(\lvert\mathcal J_t\rvert\le \frac{3}{\rho_t}\) imply
\begin{align*}
\frac18
&\le
\Pr_\gamma\!\left[
\max_{k\in N_{G'}(j)}
\langle\gamma,u_{k\mid j}\rangle
\ge \sqrt3\,t
\right]\\
&\le
\sum_{I\in\mathcal J_t}
\Pr_\gamma\!\left[
\max_{k\in N_j(I)}
\langle\gamma,u_{k\mid j}\rangle
\ge \sqrt3\,t
\right]\\
&\le
\frac{3}{\rho_t}
\max_{I\in\mathcal J_t}
\Pr_\gamma\!\left[
\max_{k\in N_j(I)}
\langle\gamma,u_{k\mid j}\rangle
\ge \sqrt3\,t
\right].
\end{align*}
Consequently, there exists an interval
\(J_j\in\mathcal J_t\) such that
\begin{align*}
\Pr_\gamma\!\left[\max_{k\in N_j(J_j)}\langle\gamma,u_{k\mid j}\rangle\ge\sqrt3\,t\right]
&\ge\frac{\rho_t}{24}.
\end{align*}

Then we group the vertices $j \in S^{(1)}$ according to their selected interval. For each \(I\in\mathcal J_t\), define \(S^{(1)}(I) := \left\{ j\in S^{(1)}:J_j=I \right\}\). By condition (ii) in Lemma~\ref{lem:kms-core} and a union
bound, we have
\begin{align*}
\frac18
&\le
\Pr_\gamma\!\left[
\max_{j\in S^{(1)}}
\langle\gamma,u_{j\mid i}\rangle
\ge \sqrt3\,t
\right]\\
&\le
\sum_{I\in\mathcal J_t}
\Pr_\gamma\!\left[
\max_{j\in S^{(1)}(I)}
\langle\gamma,u_{j\mid i}\rangle
\ge \sqrt3\,t
\right]\\
&\le
\frac3{\rho_t}
\max_{I\in\mathcal J_t}
\Pr_\gamma\!\left[
\max_{j\in S^{(1)}(I)}
\langle\gamma,u_{j\mid i}\rangle
\ge \sqrt3\,t
\right].
\end{align*}
Consequently, there exists \(J\in\mathcal J_t\) such that
\begin{equation}\label{eq:tilde-S-1-cover-property}
\Pr_\gamma\!\left[\max_{j\in S^{(1)}(J)}\langle\gamma,u_{j\mid i}\rangle\ge\sqrt3\,t\right]
\ge\frac{\rho_t}{24}.
\end{equation}
Set $\widetilde S^{(1)}:=S^{(1)}(J)$. Note that~\eqref{eq:tilde-S-1-cover-property} shows that $\left\{
u_{j\mid i}:j\in\widetilde S^{(1)}
\right\}$ is a
\(\left(\sqrt3\,t,\frac{\rho_t}{24}\right)\)-cover, and by construction, every \(j\in\widetilde S^{(1)}\) has the same selected interval \(J_j=J\).
With this choice of \(J\), we next show that \(S^{(2)}\) is the claimed recursive Gaussian certificate.

Let \(y\) be the midpoint of the interval \(J\). For every \(k\in S^{(2)}\), there exists \(j\in\widetilde S^{(1)}\) such that \(k\in N_j(J)\), and hence \(y_{jk}\in J\). Consequently, \(y_{jk}=y+o(1)\) and
\begin{align*}
\langle v_i,v_k\rangle&=\frac14+\frac34y_{jk}=b_2(y)+o(1).
\end{align*}
It remains to show that $S^{(2)}$ is a $\left(\bigl(\lambda_2(y)-\tau\bigr)t-o(t), \exp\bigl(-o(t^2)\bigr) \right)$-cover. For every \(k\in S^{(2)}\), let \(j\in\widetilde S^{(1)}\) be such that \(k\in N_j(J)\), and consider \(u_{k\mid i,j} := u_{k\mid j}-y_{jk}u_{i\mid j}\). By definition, $
u_{k\mid i,j}\perp u_{i\mid j},$
$u_{k\mid i,j}\perp u_{j\mid i}$, and
$\|u_{k\mid i,j}\|^2= 1- y_{jk}^2 = 1-y^2+o(1).$
Since \(u_{k\mid i,j}\) is obtained by removing the
\(u_{i\mid j}\)-component of \(u_{k\mid j}\),
Lemma~\ref{lem:projection}, applied to the family
\(\{u_{k\mid j}:k\in N_j(J)\}\) with
\(\varepsilon=\Phi^{-1}(1-\rho_t/96)=o(t)\), gives that $\left\{
u_{k\mid i,j}:k\in N_j(J)
\right\}$ form a $\left(
\sqrt3\,t-o(t),
\frac{\rho_t}{48}
\right)$-cover.
Thus, for every
\(j\in\widetilde S^{(1)}\), we have
\begin{equation}\label{eq:union-bound-2-level}
    \begin{aligned}
        \frac{\rho_t}{48}
&\le
\Pr_{\gamma}\left[
\max_{k\in N_j(J)}
\langle \gamma,u_{k\mid i,j}\rangle
\ge \sqrt{3}\,t-o(t)
\right]
\\
&\le
\Delta(G)
\exp\left(
-\frac{(\sqrt{3}\,t-o(t))^2}
{2(1-y^2+o(1))}
\right),
    \end{aligned}
\end{equation}
where the last inequality follows from the union bound and the upper bound in Fact~\ref{fact:tail-bound}. We know that
\begin{equation}\label{eq:log-delta-bound}
    \begin{aligned}
\log \Delta(G)
&\leq \frac{3(1+c)}{5+3c}\log n \\
&\leq \frac{3(1+c)}{5+3c}
\left(\frac{5+3c}{2}t^2 + o(t^2)\right) \\
&= \frac{3(1+c)}{2}t^2 + o(t^2),
\end{aligned}
\end{equation}
where the second line follows from the choice of \(t\) such that \(\Phi^c(t)=n^{-1/(5+3c)}\). Then, substituting \eqref{eq:log-delta-bound} into \eqref{eq:union-bound-2-level} and taking logarithms, we have
\begin{align*}
    \frac{3(1+c)}2-\frac{3}{2(1-y^2)}+o(1) \geq 0.
\end{align*}
Hence, \(y^2\le\frac{c}{1+c}+o(1)\). Projecting $y$ onto $\left[
-\sqrt{\frac{c}{1+c}},
\sqrt{\frac{c}{1+c}}
\right]$ only introduces \(o(1)\) error, thus we can assume that \(y\in \left[ -\sqrt{\frac{c}{1+c}}, \sqrt{\frac{c}{1+c}} \right]\) while preserving $y_{jk}=y+o(1)$ for all \(k\in S^{(2)}\).

Now, by direct substitution, we have
\begin{align*}
v_k&=\left(\frac14+\frac34y_{jk}\right)v_i
-\frac{\sqrt3}{4}(1-y_{jk})u_{j\mid i}
+\frac{\sqrt3}{2}u_{k\mid i,j},
\end{align*}
which gives
\[
v_k-\langle v_i,v_k\rangle v_i
=
-\frac{\sqrt{3}}{4}(1-y_{jk})u_{j\mid i}
+
\frac{\sqrt{3}}{2}u_{k\mid i,j}.
\]
By symmetry, $
\left\{
-u_{j\mid i}:j\in\widetilde S^{(1)}
\right\}$ form a $\left(
\sqrt3\,t-o(t),
\exp(-o(t^2))
\right)$-cover. As established previously, $\left\{
u_{k\mid i,j}:k\in N_j(J)
\right\}$ form a $\left(
\sqrt3\,t-o(t),
\frac{\rho_t}{48}
\right)$-cover, which is also a $\left(
\sqrt3\,t-o(t), \exp(-o(t^2))
\right)$-cover. We also have $\|u_{j\mid i}\| = 1$ and $\|u_{k\mid i,j}\|^2 = 1-y^2+o(1)$. 
Moreover,
\begin{align*}
\log|\widetilde S^{(1)}|\le \log n\leq \frac{5+3c}{2}t^2+o(t^2),
\quad \log|N_j(J)|\le \log \Delta(G) \le \frac{3(1+c)}2t^2+o(t^2).
\end{align*} 
And since \(y\in[-\sqrt{c/(1+c)},\sqrt{c/(1+c)}]\), we have
\begin{align*}
E-\lambda^2=2+3c>0,\quad U^2-\frac1{1+c}=\frac{c}{1+c}-y^2\ge0.
\end{align*}
We now apply Lemma~\ref{lem:orthogonal-cover-composition}, associating
\(x=-u_{j\mid i}\) with
\(Y_x=\{u_{k\mid i,j}:k\in N_j(J)\}\), and using parameters
\(\lambda=\sqrt3\), \(E=5+3c\), and \(U^2=1-y^2\). With probability at least
\(
\exp\bigl(-o(t^2)\bigr),
\)
there is some retained pair $(j,k)$ such that simultaneously
\[
\langle \gamma, -u_{j\mid i}\rangle \geq \sqrt{3}\,t-o(t),\quad 
\langle \gamma, u_{k\mid i,j}\rangle
\geq
\left(
\sqrt{3}
-
\sqrt{(2+3c)
\left(\frac{c}{1+c}-y^2\right)}
\right)t-o(t).
\]
Under this event, we have
\[
\left\langle \gamma, v_k-\langle v_i,v_k\rangle v_i \right\rangle
\geq
\left[
\frac{9}{4}
-\frac{3}{4}y
-\frac{\sqrt{3}}{2}
\sqrt{(2+3c)\left(\frac{c}{1+c}-y^2\right)}
\right]t-o(t).
\]
Therefore, lowering the threshold by \(\tau t\) gives
\begin{align*}
\Pr_\gamma\left[\max_{k\in S^{(2)}}\left\langle\gamma,
\frac{v_k-\langle v_i,v_k\rangle v_i}{\sqrt{1-\langle v_i,v_k\rangle^2}}
\right\rangle\ge\bigl(\lambda_2(y)-\tau\bigr)t-o(t)\right]
&\ge\exp(-o(t^2)).
\end{align*}
Thus, \(S^{(2)}\) is the claimed recursive Gaussian certificate.
\end{proof}

The same construction extends a level-\(\ell\) certificate to level
\(\ell+1\) by appending one edge to each retained walk. Suppose the root is $i$ and the level-\(\ell\) certificate set is \(S^{(\ell)}\). For each \(j\in S^{(\ell)}\), we choose a neighbor \(k\in N_{G'}(j)\) to produce a length-\((\ell+1)\) walk $i\rightarrow \cdots \rightarrow j\rightarrow k.$ We define the \emph{type} of the walk \(i\to \cdots \to j\to k\) by \[y_{jk}:=\left\langle u_{k\mid j},\frac{v_i-\langle v_i,v_j\rangle v_j }{\sqrt{1 - \langle v_i,v_j\rangle^2}}\right\rangle.\]
The type characterizes how the correlation with the root $v_i$ changes as we extend the walk. Direct computation gives
\begin{align*}
\langle v_i,v_k\rangle
&=
-\frac{\langle v_i,v_j\rangle}{2}
+\frac{\sqrt3}{2}\sqrt{1-\langle v_i,v_j\rangle^2}\,y_{jk}.
\end{align*}
Define the following functions \(A(a,y)\) and \(B(a,y)\) for \(a\in(-1,1)\) and
\(y\in[-1,1]\):
\begin{align}
B(a,y)&:=-\frac a2+\frac{\sqrt3}{2}\sqrt{1-a^2}\,y,\\
A(a,y)&:=\frac12\sqrt{1-a^2}+\frac{\sqrt3}{2}ay.
\label{eq:BA}
\end{align}
If \(a=\langle v_i,v_j\rangle\) and \(y=y_{jk}\), then 
\[
v_k
=
B(a,y)v_i
-
A(a,y)\left(
\frac{v_j-\langle v_i,v_j\rangle v_i}
{\sqrt{1-\langle v_i,v_j\rangle^2}}\right)
+
\frac{\sqrt{3}}{2}\left(u_{k\mid j}
-
y_{jk}
\frac{v_i-\langle v_i,v_j\rangle v_j}
{\sqrt{1-\langle v_i,v_j\rangle^2}}\right)
\]
Recall that for level-\(1\) certificate, these reduce to
\(B(-1/2,y)=b_2(y)\) and \(A(-1/2,y)=\frac{\sqrt3}{4}(1-y)\).
For extending the level-\(\ell\) certificate, set \(L:=\frac12\sqrt{(1+4c)/(1+c)}\) and
\begin{align}
\nu_0(c)&:=\min\left\{
\frac{\sqrt3}{4(1+c)}\left(\sqrt{1+c}-\sqrt{c(1+4c)}\right),
\frac32-\frac{\sqrt3}{2}\sqrt{\frac{c(5+3c)}{1+c}}
\right\}>0.
\label{eq:invariant-constants}
\end{align}
The positivity follows from \(0<c<c^\dagger<1/2\). 

Now we are ready to state the main lemma that extends a level-\(\ell\) certificate to level \(\ell+1\).
\begin{lemma}[Extending a recursive Gaussian certificate]
\label{lem:recursive-step}
Fix \(0<c<c^\dagger\) and \(0<\tau<\nu_0(c)/2\). Let
\(S^{(\ell)}\) be a level-\(\ell\) recursive Gaussian certificate rooted
at \(i\) with parameters \(\bigl(a,\lambda t-o(t),\exp(-o(t^2))\bigr)\),
where \(|a|\le L\) and \(0\le\lambda\le\sqrt{5+3c}\).

Then, for some \(y\in[-\sqrt{c/(1+c)},\sqrt{c/(1+c)}]\), there exists a level-\((\ell+1)\) recursive Gaussian certificate \(S^{(\ell+1)}\) rooted at \(i\) with parameters
\(\bigl(a_{\mathrm{new}},\lambda_{\mathrm{new}}t-o(t),\exp(-o(t^2))\bigr)\), where
\begin{align*}
a_{\mathrm{new}}&=B(a,y),\\
\lambda_{\mathrm{new}}
&=\frac{\lambda A(a,y)+\frac32
-\frac{\sqrt3}{2}\sqrt{(5+3c-\lambda^2)\left(\frac{c}{1+c}-y^2\right)}}
{\sqrt{1-B(a,y)^2}}-\tau.
\end{align*}
\end{lemma}

Observe that taking \(a=-1/2\) and \(\lambda=\sqrt3\) gives
\(a_{\mathrm{new}}=b_2(y)\) and \(\lambda_{\mathrm{new}}=\lambda_2(y)-\tau\), recovering the transition in Lemma~\ref{lem:level-2-certificate}. The proof of Lemma \ref{lem:recursive-step} follows from the same interval selection, projection, and composition steps as in Lemma \ref{lem:level-2-certificate}. We defer the proof to Appendix~\ref{app:recursive-step}.

\section{The Potential Function}\label{sec:potential-function}

Suppose that the KMS\(^{\prime}\) failure subgraph exists. Starting from
the level-\(1\) certificate, Lemma~\ref{lem:recursive-step} then constructs
certificates at successively higher walk levels. A level-\(\ell\)
certificate with parameters \((a,\lambda)\) has vertices with correlation \(a+o(1)\) with the root and their normalized directions satisfy
\[
\Pr_\gamma\left[
\max_{j\in S^{(\ell)}}
\left\langle
\gamma,
\frac{v_j-\langle v_i,v_j\rangle v_i}
{\sqrt{1-\langle v_i,v_j\rangle^2}}
\right\rangle
\ge\lambda t-o(t)
\right]
\ge\exp(-o(t^2)).
\]
Every certificate contains at most \(n\) distinct unit vectors. Consequently, a union bound gives
\begin{align*}
\exp(-o(t^2))
&\le
\Pr_\gamma\left[
\max_{j\in S^{(\ell)}}
\left\langle
\gamma,
\frac{v_j-\langle v_i,v_j\rangle v_i}
{\sqrt{1-\langle v_i,v_j\rangle^2}}
\right\rangle
\ge\lambda t-o(t)
\right]\\
&\le
n\Phi^c(\lambda t-o(t))\\
&=
\exp\left(
\frac{5 + 3c -\lambda^2}{2}t^2+o(t^2)
\right).
\end{align*}
Thus it suffices to find a certificate with \(\lambda>\sqrt{5 + 3c} +\Omega(1)\). We achieve this by exhibiting a potential function with uniform increase at each step. Recall that each recursive Gaussian certificate is tracked by  $(a,s,\delta)$. Writing $s=\lambda t-o(t),$
the potential function we use is
\[
V(a,\lambda)
:=
\lambda
-
\frac{3(1+c)}{1+4c}
\left(
\frac{1}{\sqrt{1-a^2}}-1
\right).
\]
The asymptotic form of the probability parameter \(\delta=\exp(-o(t^2))\) is preserved throughout the recursion and therefore does not enter the potential explicitly. In Subsection~\ref{subsec:properties}, we will motivate the function and discuss its derivation. We will then use it to complete the proof of Theorem \ref{thm:sparse-progress}.

\subsection{Motivating the function} \label{subsec:properties}

In an ideal path of certificates, we would like the threshold coefficient \(\lambda\) to increase by $\Omega(1)$ at each step. We examine why this direct approach fails. Let \(r_{jk}\) denote the orthogonal projection of \(u_{k\mid j}\) onto \(\operatorname{span}\{v_i,v_j\}^{\perp}\). For the purpose of this subsection, we shall suppress the little-\(o\) errors and the fixed coefficient loss \(\tau\) in Lemma~\ref{lem:recursive-step}. Consider the ideal state transition $(a,\lambda)
\longmapsto
\bigl(B(a,y),\lambda_{\mathrm{new}}\bigr)$
with the new endpoint direction satisfying
\begin{equation}
v_k-\langle v_i,v_k\rangle v_i
=
-A(a,y)x_j+\frac{\sqrt3}{2}r_{jk}+o(1),
\end{equation}
and therefore
\begin{align}
\lambda_{\mathrm{new}}-\lambda
&=
\frac{
\lambda A(a,y)+\frac32
-\frac{\sqrt3}{2}
\sqrt{(5+3c-\lambda^2)\left(\frac{c}{1+c}-y^2\right)}
}{
\sqrt{1-B(a,y)^2}
}
-\lambda
\notag\\
&=
\frac{
\frac32
-
\left(
\sqrt{1-B(a,y)^2}-A(a,y)
\right)\lambda
-
\frac{\sqrt3}{2}
\sqrt{(5+3c-\lambda^2)\left(\frac{c}{1+c}-y^2\right)}
}{
\sqrt{1-B(a,y)^2}
}.
\label{eq:lambda-difference}
\end{align}

The preceding identity shows why \(\lambda\) need not increase: the positive contribution $\tfrac32$ can be offset by both the composition loss and the change in normalization. We therefore look for a correction depending on the correlation and consider the potential function
\[
V_\varphi(a,\lambda):=\lambda-\varphi(a).
\]
Writing 
\[a_{\mathrm{new}}
=
B(a,y)
=
-\frac a2
+
\frac{\sqrt3}{2}\sqrt{1-a^2}\,y, \]
define, for an arbitrary function \(\varphi\),
\begin{align*}
H_\varphi(a,a_{\mathrm{new}})
&:=
\frac32
+
\sqrt{1-a_{\mathrm{new}}^2}
\left(
\varphi(a)-\varphi(a_{\mathrm{new}})
\right),\\
W(a,a_{\mathrm{new}})
&:=
\left(
\sqrt{1-a_{\mathrm{new}}^2}-A(a,y)
\right)^2
+
\frac34\left(\frac{c}{1+c}-y^2\right).
\end{align*}
The potential change satisfies
\begin{align*}
\sqrt{1-a_{\mathrm{new}}^2}\cdot
\left(
V_\varphi(a_{\mathrm{new}},\lambda_{\mathrm{new}})
-
V_\varphi(a,\lambda)
\right)
&=
\sqrt{1-a_{\mathrm{new}}^2}
\left(
\lambda_{\mathrm{new}}-\lambda
\right)
+
\sqrt{1-a_{\mathrm{new}}^2}
\left(
\varphi(a)-\varphi(a_{\mathrm{new}})
\right)\\
&=
\frac32
-
\left(
\sqrt{1-a_{\mathrm{new}}^2}-A(a,y)
\right)\lambda\\
&\qquad\quad
-
\frac{\sqrt3}{2}
\sqrt{\frac{c}{1+c}-y^2}\sqrt{5+3c-\lambda^2}\\
&\qquad\quad
+
\sqrt{1-a_{\mathrm{new}}^2}
\left(
\varphi(a)-\varphi(a_{\mathrm{new}})
\right)\\
&\ge
H_\varphi(a,a_{\mathrm{new}})
-
\sqrt{(5+3c)W(a,a_{\mathrm{new}})}
\end{align*}
where the inequality follows from Cauchy--Schwarz.
Thus we first seek a function \(\varphi\) satisfying
\[
H_\varphi(a,a_{\mathrm{new}})
\ge
\sqrt{(5+3c)W(a,a_{\mathrm{new}})}
\]
for all possible transitions.

To guess such a function \(\varphi\), we examine the worst case change in the correlation: the two-cycle given by $a_{\mathrm{new}}=-a$ (which is feasible when $a^2\le3c/(1+4c)$). Along this correlation pattern, direct substitution gives $W=\frac{1+4c}{4(1+c)}.$
Moreover, because the two directions \(a\to-a\) and \(-a\to a\) are symmetric, we seek $\varphi(a)=\varphi(-a)$, which gives $H_\varphi(a,-a)=\frac32.$ Therefore, on this correlation pattern, we have:
\[
H_\varphi(a,-a)^2
=
\frac{9(1+c)}{1+4c}W(a,-a).
\]
Accordingly, define 
\[
Q_\varphi(a,a_{\mathrm{new}})
:=
H_\varphi(a,a_{\mathrm{new}})^2
-
\frac{9(1+c)}{1+4c}W(a,a_{\mathrm{new}}).
\]
Since we seek \(Q_\varphi(a,a_{\mathrm{new}})\ge0\) for every feasible transition, with equality at \(a_{\mathrm{new}}=-a\), we use the first-order stationarity condition here to guess a candidate function $\varphi$. By a direct computation, this gives the condition

\begin{align*}
0
&=
\left.
\partial_{a_{\mathrm{new}}}
Q_\varphi(a,a_{\mathrm{new}})
\right|_{a_{\mathrm{new}}=-a}
\\
&=
\left.
\partial_{a_{\mathrm{new}}}
\left(
H_\varphi(a,a_{\mathrm{new}})^2
-
\frac{9(1+c)}{1+4c}
W(a,a_{\mathrm{new}})
\right)
\right|_{a_{\mathrm{new}}=-a}
\\
&=
2H_\varphi(a,-a)
\left.
\partial_{a_{\mathrm{new}}}
H_\varphi(a,a_{\mathrm{new}})
\right|_{a_{\mathrm{new}}=-a}
-
\frac{9(1+c)}{1+4c}
\left.
\partial_{a_{\mathrm{new}}}
W(a,a_{\mathrm{new}})
\right|_{a_{\mathrm{new}}=-a}
\\
&=
3
\left.
\partial_{a_{\mathrm{new}}}
H_\varphi(a,a_{\mathrm{new}})
\right|_{a_{\mathrm{new}}=-a}
-
\frac{9(1+c)}{1+4c}
\left.
\partial_{a_{\mathrm{new}}}
W(a,a_{\mathrm{new}})
\right|_{a_{\mathrm{new}}=-a}
\\
&=
3\sqrt{1-a^2}\,\varphi'(a)
-
\frac{9(1+c)}{1+4c}
\frac{a}{1-a^2}.
\end{align*}
Rearranging gives
\[
\varphi'(a)
=
\frac{3(1+c)}{1+4c}
\frac{a}{(1-a^2)^{3/2}}.
\]
With the normalization $\varphi(0)=0$, integration gives
\[
\varphi(a)
=
\frac{3(1+c)}{1+4c}
\left(
\frac{1}{\sqrt{1-a^2}}-1
\right).
\]
Plugging in this $\varphi$ gives the $V(a, \lambda)$ we utilize in our proof. The following lemma verifies that the candidate obtained from these equality conditions satisfies the required inequality for every feasible transition.

Define
\[
\mu(c)
:=
\frac{3}{8(1+c)}
\left(
3\sqrt{\frac{1+c}{1+4c}}-\sqrt{5+3c}
\right)
=
\frac{3(c^\dagger-c)(12c+\sqrt{97}+7)}
{8(1+c)(1+4c)\left(3\sqrt{\frac{1+c}{1+4c}}+\sqrt{5+3c}\right)}
>0.
\]
The final inequality holds because \(c<c^\dagger\). Recalling \(L\) from
\eqref{eq:invariant-constants}, write
\[
\mathcal K:=[-L,L]\times[0,\sqrt{5+3c}].
\]

\begin{lemma}
\label{lem:potential-properties}

Let \((a,\lambda)\in\mathcal K\) be the parameters of a recursive
Gaussian certificate at level $\ell$ and
\((a_{\mathrm{new}},\lambda_{\mathrm{new}})\) the parameters of the new certificate obtained from Lemma~\ref{lem:recursive-step} with loss parameter \(\tau\le\mu(c)/2\). Then:
\begin{enumerate}[(i)]
    \item $|a_{\mathrm{new}}|\le L.$

    \item $  V(a_{\mathrm{new}},\lambda_{\mathrm{new}})
    -
    V(a,\lambda)
    \ge
    \mu(c)-\tau.$
\end{enumerate}
\end{lemma}

\begin{proof}
For the first item, direct substitution gives
\begin{align*}
A(a,y)
&\ge
\frac12\sqrt{1-L^2}
-\frac{\sqrt3}{2}L\sqrt{\frac{c}{1+c}}
>0,
\\
A(a,y)^2+|a_{\mathrm{new}}|^2
&=
A(a,y)^2+B(a,y)^2\\
&=
\frac14+\frac34y^2
\le L^2,
\end{align*}
which implies \(|a_{\mathrm{new}}|\le L\). For the second item, substituting
\(
y=(2a_{\mathrm{new}}+a)/(\sqrt3\sqrt{1-a^2})
\)
into the definitions of \(H_\varphi\) and \(W\) gives
\begin{align*}
&H_\varphi(a,a_{\mathrm{new}})^2
-\frac{9(1+c)}{1+4c}W(a,a_{\mathrm{new}})\\
&=
\frac{9(1+c)}{(1-a^2)(1+4c)^2}
\left(
1+aa_{\mathrm{new}}
-\sqrt{(1-a^2)(1-a_{\mathrm{new}}^2)}
\right)\\
&\qquad\cdot
\Bigg[
2+8c+6caa_{\mathrm{new}}
-(1+7c)
\left(
1+aa_{\mathrm{new}}
-\sqrt{(1-a^2)(1-a_{\mathrm{new}}^2)}
\right)
\Bigg]\\
&\ge0.
\end{align*}
Indeed, the first parenthesis is nonnegative because
\(
(1+aa_{\mathrm{new}})^2
-(1-a^2)(1-a_{\mathrm{new}}^2)
 =
(a+a_{\mathrm{new}})^2
\),
and it is at most \(1+aa_{\mathrm{new}}\). Thus, the final bracket is
at least \((1+c)(1-aa_{\mathrm{new}})>0\). Moreover,
\begin{align*}
H_\varphi(a,a_{\mathrm{new}})
&=
\frac32-\frac{3(1+c)}{1+4c}
+\frac{3(1+c)}{1+4c}
\frac{\sqrt{1-a_{\mathrm{new}}^2}}{\sqrt{1-a^2}}\\
&\ge
\frac{3}{2(1+4c)}
\left(-1+2c+\sqrt{3(1+c)}\right)>0.
\end{align*}
Therefore
\(
H_\varphi(a,a_{\mathrm{new}})
\ge
3\sqrt{(1+c)/(1+4c)}\sqrt{W(a,a_{\mathrm{new}})}
\).
Applying this gives
\begin{align*}
&V(a_{\mathrm{new}},\lambda_{\mathrm{new}})
-
V(a,\lambda)\\
&=
\lambda_{\mathrm{new}}-\lambda
-
\frac{3(1+c)}{1+4c}
\left(
\frac1{\sqrt{1-a_{\mathrm{new}}^2}}
-
\frac1{\sqrt{1-a^2}}
\right)\\
&\ge
\frac1{\sqrt{1-a_{\mathrm{new}}^2}}
\Bigg[
\frac32-\frac{3(1+c)}{1+4c}
+
\frac{3(1+c)}{1+4c}
\frac{\sqrt{1-a_{\mathrm{new}}^2}}{\sqrt{1-a^2}}\\
&\hspace{2.5cm}
-
\left(
\sqrt{1-a_{\mathrm{new}}^2}-A(a,y)
\right)\lambda
-
\frac{\sqrt3}{2}
\sqrt{\frac{c}{1+c}-y^2}\sqrt{5+3c-\lambda^2}
\Bigg]
-\tau\\
&\ge
\frac1{\sqrt{1-a_{\mathrm{new}}^2}}
\Bigg[
\frac32-\frac{3(1+c)}{1+4c}
+
\frac{3(1+c)}{1+4c}
\frac{\sqrt{1-a_{\mathrm{new}}^2}}{\sqrt{1-a^2}}\\
&\hspace{2.5cm}
-
\sqrt{
(5+3c)\left[
\left(
\sqrt{1-a_{\mathrm{new}}^2}-A(a,y)
\right)^2
+\frac34\left(\frac{c}{1+c}-y^2\right)
\right]
}
\Bigg]
-\tau\\
&\ge
\frac{
3\sqrt{\frac{1+c}{1+4c}}-\sqrt{5+3c}
}{
\sqrt{1-a_{\mathrm{new}}^2}
}
\sqrt{
\left(
\sqrt{1-a_{\mathrm{new}}^2}-A(a,y)
\right)^2
+\frac34\left(\frac{c}{1+c}-y^2\right)
}
-\tau\\
&\ge
\frac{
3\sqrt{\frac{1+c}{1+4c}}-\sqrt{5+3c}
}{
\sqrt{1-a_{\mathrm{new}}^2}
}
\left|
\sqrt{1-a_{\mathrm{new}}^2}-A(a,y)
\right|
-\tau\\
&\ge
\frac{
3\sqrt{\frac{1+c}{1+4c}}-\sqrt{5+3c}
}{
\sqrt{1-a_{\mathrm{new}}^2}
}
\left(
\sqrt{1-a_{\mathrm{new}}^2}-A(a,y)
\right)
-\tau\\
&=
\frac{
3\sqrt{\frac{1+c}{1+4c}}-\sqrt{5+3c}
}{
\sqrt{1-a_{\mathrm{new}}^2}
}
\frac{
1-a_{\mathrm{new}}^2-A(a,y)^2
}{
\sqrt{1-a_{\mathrm{new}}^2}+A(a,y)
}
-\tau\\
&\ge
\frac{
3\sqrt{\frac{1+c}{1+4c}}-\sqrt{5+3c}
}{
\sqrt{1-a_{\mathrm{new}}^2}
}
\frac{1-L^2}{2}
-\tau\\
&=
\frac{\mu(c)}{\sqrt{1-a_{\mathrm{new}}^2}}
-\tau\\
&\ge
\mu(c)-\tau.
\end{align*}
\end{proof}

\subsection{Proof of Theorem \ref{thm:sparse-progress}}\label{subsec:proof-of-main}

We prove Theorem \ref{thm:sparse-progress} by repeatedly applying Lemma \ref{lem:recursive-step} and showing that the potential $V$ increases at every step. This rules out the bad subgraph $G'$ from Lemma \ref{lem:kms-core}, so repeated \kmsp trials yield a large independent set with high probability.

\begin{proof}[Proof of Theorem \ref{thm:sparse-progress}]
Let $\mu(c)$, $\mathcal{K}$, and $V$ be defined as in
Lemma~\ref{lem:potential-properties}, and let $\nu_0(c)$ be defined as in
\eqref{eq:invariant-constants}. Define
\[R_V:=\sup_{\mathcal K}V-\inf_{\mathcal K}V.\] Since \(\mathcal K\) is compact and \(V\) is continuous on
\(\mathcal K\), we have \(R_V<\infty\). Choose an integer
\(K>2R_V/\mu(c)+2\) and constants
\(0<\tau<\min\{\nu_0(c)/2,\mu(c)/8\}\) and
\(0<\delta<\mu(c)/8\).

Now, recall that the threshold \(t\) is chosen so that $\Phi^c(t)=n^{-1/(5+3c)}.$ Suppose, for the sake of contradiction, that alternative~\emph{(ii)} of Lemma~\ref{lem:kms-core} holds, and let \(G'=(V',E')\) be the resulting nonempty induced subgraph. Fix a root \(i\in V'\). Then $S^{(1)}:=N_{G'}(i)$ is a level-\(1\)
$\left(
-\frac12,
\sqrt3\,t,
\frac18
\right)$-recursive Gaussian certificate rooted at \(i\). 

Indeed, write $(a_1,\lambda_1):=\left(-\frac12,\sqrt3\right),$ so that the set \(S^{(1)}\) is an
\(\left(a_1,\lambda_1t-o(t),\exp(-o(t^2))\right)\)
recursive Gaussian certificate, and
\((a_1,\lambda_1)\in\mathcal K\). We now apply Lemma~\ref{lem:recursive-step} repeatedly. Suppose that at
level \(\ell\) we have a
$\left(
a_\ell,
\lambda_\ell t-o(t),
\exp(-o(t^2))
\right)$
recursive Gaussian certificate \(S^{(\ell)}\), where \((a_\ell,\lambda_\ell)\in\mathcal K\). Lemma~\ref{lem:recursive-step}
produces a level-\((\ell+1)\)
\(\left(a_{\ell+1},\lambda_{\ell+1}t-o(t),\exp(-o(t^2))\right)\)
recursive Gaussian certificate. We consider two cases according to the value of
\(\lambda_{\ell+1}\).

\paragraph{Case 1: \(\lambda_{\ell+1}>\sqrt{5+3c}+\delta\).}
For all sufficiently large \(n\), we have
\(\lambda_{\ell+1}t-o(t)\ge(\sqrt{5+3c}+\delta/2)t\).
The certificate contains at most \(n\) distinct unit vectors. Therefore,
using
\(\log n=\frac{5+3c}{2}t^2+o(t^2)\) and the Gaussian tail bound, we obtain
\begin{align*}
\Pr_\gamma\left[
\max_{j\in S^{(\ell+1)}}
\langle\gamma,x_j\rangle
\ge
\lambda_{\ell+1}t-o(t)
\right]
&\le
n\Phi^c\left(
\left(\sqrt{5+3c}+\frac\delta2\right)t
\right)\\
&=
\exp\left(-\Omega_{c,\delta}(t^2)\right).
\end{align*}
This contradicts the \(\exp(-o(t^2))\) cover lower bound of the
level-\((\ell+1)\) certificate. Hence Case~1 is impossible.

\paragraph{Case 2: \(\lambda_{\ell+1}\le\sqrt{5+3c}+\delta\).}
The recursive transition and definition of \(\nu_0(c)\) give
\begin{align*}
\lambda_{\ell+1}
&\ge
\frac{
\lambda_\ell A(a_\ell,y)+\frac32
-\frac{\sqrt3}{2}
\sqrt{
(5+3c-\lambda_\ell^2)
\left(\frac{c}{1+c}-y^2\right)
}
}{
\sqrt{1-a_{\ell+1}^2}
}
-\tau\\
&\ge
\frac{
\frac32-\frac{\sqrt3}{2}
\sqrt{\frac{c(5+3c)}{1+c}}
}{
\sqrt{1-a_{\ell+1}^2}
}
-\tau\\
&\ge
\nu_0(c)-\tau
>0.
\end{align*}
Now, define $\lambda'_{\ell+1}
:=
\min\left\{
\lambda_{\ell+1},
\sqrt{5+3c}
\right\}.$
Lowering the represented threshold preserves the cover property,
so \(S^{(\ell+1)}\) is also an
\(\left(a_{\ell+1},\lambda'_{\ell+1}t-o(t),\exp(-o(t^2))\right)\)
recursive Gaussian certificate. Part~\textnormal{(i)} of
Lemma~\ref{lem:potential-properties} gives $|a_{\ell+1}|\le L,$ and hence $(a_{\ell+1},\lambda'_{\ell+1})\in\mathcal K.$
Before lowering the threshold, part~\textnormal{(ii)} of
Lemma~\ref{lem:potential-properties} gives
\(V(a_{\ell+1},\lambda_{\ell+1})-V(a_\ell,\lambda_\ell)
\ge\mu(c)-\tau\). Since \(V\) is linear in its second argument with coefficient one,
\begin{align*}
V(a_{\ell+1},\lambda'_{\ell+1})
-
V(a_\ell,\lambda_\ell)
&=
V(a_{\ell+1},\lambda_{\ell+1})
-
V(a_\ell,\lambda_\ell)
-
\left(
\lambda_{\ell+1}-\lambda'_{\ell+1}
\right)\\
&\ge
\mu(c)-\tau
-
\left(
\lambda_{\ell+1}-\lambda'_{\ell+1}
\right)\\
&\ge \mu(c)-\tau-\delta > \frac{\mu(c)}2.
\end{align*}

Thus, unless Case~1 gives an immediate contradiction, we may replace the
level-\((\ell+1)\) certificate by the lowered certificate above and
continue the recursion from the state $(a_{\ell+1},\lambda'_{\ell+1})\in\mathcal K,$
whose potential exceeds that of the preceding state by more than
\(\mu(c)/2\). For the next iteration, relabel
\(\lambda'_{\ell+1}\) as \(\lambda_{\ell+1}\).

Iterating this argument, either Case~1 occurs at some step, giving an
immediate contradiction, or we obtain \(K\) successive states in
\(\mathcal K\), with each transition increasing the potential by more than
\(\mu(c)/2\). Since \(K\) is fixed, we may take \(n\) sufficiently large
that all asymptotic estimates used in the first \(K-1\) recursive steps
hold simultaneously. In the latter case,
\begin{align*}
V(a_K,\lambda_K)-V(a_1,\lambda_1)
&> \frac{(K-1)\mu(c)}2 > R_V,
\end{align*}
contradicting the definition of \(R_V\). Therefore
alternative~\emph{(ii)} of Lemma~\ref{lem:kms-core} is impossible.

It therefore follows that one KMS\(^{\prime}\) trial satisfies $\E|I|
\ge
\frac14 n^{(4+3c)/(5+3c)}.$ By Markov's inequality applied to \(n-|I|\),
\begin{align*}
\Pr\left[
|I|<
\frac18 n^{(4+3c)/(5+3c)}
\right]
&=
\Pr\left[
n-|I|>
n-\frac18 n^{(4+3c)/(5+3c)}
\right]\\
&\le
\frac{n-\E|I|}
{n-\frac18 n^{(4+3c)/(5+3c)}}.
\end{align*}
Therefore, $\Pr\left[
|I|\ge
\frac18 n^{(4+3c)/(5+3c)}
\right]
\ge
\frac18 n^{-1/(5+3c)}.$ Running $O \left(n^{1/(5+3c)}\log n\right)$ independent KMS\(^{\prime}\) trials and retaining the largest returned set gives the result in polynomial time with high probability.
\end{proof}

\section*{Acknowledgments}

\paragraph{Statement of AI Use.} The paper was motivated by the work of Bansal, Huang, and Lee \cite{BHL26}. The authors initially asked ChatGPT 5.6 Sol whether their analysis could be extended to third- and fourth-order neighborhoods. This led to an improved analysis of the work of \cite{ACC06} that relied heavily on numerical certificates. Through a subsequent series of interactions with ChatGPT 5.6 Sol, this approach was extended to a fully analytic proof that utilized arbitrary depth level-\(L\) neighborhoods. The authors then spent time understanding, verifying, and simplifying the proof. Codex assisted with typesetting.

\bibliographystyle{alpha}
\bibliography{ref}

\appendix

\section{Proof of Lemma~\ref{lem:recursive-step}}
\label{app:recursive-step}

\begin{proof}
For \(j\in S^{(\ell)}\), write \(a_j:=\langle v_i,v_j\rangle\) and
\(p_j:=\sqrt{1-a_j^2}\), and let \(x_j:=(v_j-a_jv_i)/p_j\) and
\(z_j:=(v_i-a_jv_j)/p_j\) be the normalized directions from the root to
\(j\) and from \(j\) to the root, respectively. By definition,
\[
    v_j=a_jv_i+p_jx_j, \qquad z_j=p_jv_i-a_jx_j.
\]
Because $S^{(\ell)}$ is an $\bigl(a,\lambda t-o(t),\exp(-o(t^2))\bigr)$-certificate, we have \(a_j=a+o(1)\) uniformly over \(j\). Furthermore, symmetry of covers implies that \(\{-x_j:j\in S^{(\ell)}\}\) is a \(\bigl(\lambda t-o(t),\exp(-o(t^2))\bigr)\)-cover.

Now, for an extension \(j\to k\), recall that its type is given by
\(y_{jk}:=\langle u_{k\mid j},z_j\rangle\) and its projected edge
direction by \(r_{jk}:=u_{k\mid j}-y_{jk}z_j\). Note that at level \(1\), we have
\(x_j=u_{j\mid i}\), \(z_j=u_{i\mid j}\), and
\(r_{jk}=u_{k\mid i,j}\), so these are exactly the quantities used in
Lemma \ref{lem:level-2-certificate}. 

Since \(u_{k\mid j}\perp v_j\), we have
\(r_{jk}\perp\operatorname{span}\{v_i,v_j\}\) and \(\|r_{jk}\|^2=1-y_{jk}^2\). Thus, applying \eqref{eq:edge-decomposition} to the edge \(\{j,k\}\) and substituting gives
\begin{align*}
\langle v_i,v_k\rangle
&=-\frac{a_j}{2}+\frac{\sqrt3}{2}p_jy_{jk}=B(a_j,y_{jk}),\\
v_k-\langle v_i,v_k\rangle v_i
&=-\left(\frac{p_j}{2}+\frac{\sqrt3}{2}a_jy_{jk}\right)x_j
+\frac{\sqrt3}{2}r_{jk}=-A(a_j,y_{jk})x_j+\frac{\sqrt3}{2}r_{jk}.
\end{align*}

We select a common type as in the proof of Lemma \ref{lem:level-2-certificate}.
Set \(\rho_t:=(\log t)^{-2}\), and construct a partition \(\mathcal J_t\) of \([-1,1]\) into at most \(3/\rho_t\) consecutive intervals of width at most \(\rho_t\). For \(I\in\mathcal J_t\),
let \(N_j(I):=\{k\in N_{G'}(j):y_{jk}\in I\}\). The condition (ii) in Lemma~\ref{lem:kms-core} and the same union-bound argument as in Lemma~\ref{lem:level-2-certificate} give an interval \(J_j\in\mathcal J_t\) with
\begin{align*}
\Pr_\gamma\!\left[
\max_{k\in N_j(J_j)}\langle\gamma,u_{k\mid j}\rangle\ge\sqrt3\,t
\right]&\ge\frac{\rho_t}{24}.
\end{align*}
Grouping the vertices \(j\in S^{(\ell)}\) by \(J_j\) and applying the union bound once
more gives an interval \(J\) for which
\(\widetilde S^{(\ell)}:=\{j\in S^{(\ell)}:J_j=J\}\) retains an outer
\(\bigl(\lambda t-o(t),\exp(-o(t^2))\bigr)\)-cover. 

Let \(y\) be the midpoint of \(J\), and set
\(S^{(\ell+1)}:=\bigcup_{j\in\widetilde S^{(\ell)}}N_j(J)\), with each
endpoint included only once. Uniformly over retained pairs,
\(y_{jk}=y+o(1)\), so the identities above yield
\begin{align*}
\langle v_i,v_k\rangle&=B(a,y)+o(1),\\
v_k-\langle v_i,v_k\rangle v_i
&=-\bigl(A(a,y)+o(1)\bigr)x_j+\frac{\sqrt3}{2}r_{jk}.
\end{align*}

We now verify the projected cover probability: since
\(2\Phi^c(\Phi^{-1}(1-\rho_t/96))=\rho_t/48\), Lemma~\ref{lem:projection},
applied with loss \(\Phi^{-1}(1-\rho_t/96)=o(t)\), shows that every retained family
\(\{r_{jk}:k\in N_j(J)\}\) is a
\(\bigl(\sqrt3\,t-o(t),\rho_t/48\bigr)\)-cover, uniformly over \(j\).
As in the proof of Lemma \ref{lem:level-2-certificate}, \(\|r_{jk}\|^2=1-y^2+o(1)\) and the degree
bound imply
\begin{align*}
\frac{\rho_t}{48}
&\le\Pr_\gamma\!\left[
\max_{k\in N_j(J)}\langle\gamma,r_{jk}\rangle\ge\sqrt3\,t-o(t)
\right]\\
&\le\Delta(G)\exp\left(-\frac{(\sqrt3\,t-o(t))^2}{2(1-y^2+o(1))}\right).
\end{align*}
We now verify the bound on \(y\). Taking logarithms, dividing by \(t^2\), and using
\(\log(1/\rho_t)=o(t^2)\) and
\(\log\Delta(G)\le\frac{3(1+c)}2t^2+o(t^2)\) gives
\(0\le\frac{3(1+c)}2-\frac{3}{2(1-y^2)}+o(1)\), and hence
\(y^2\le c/(1+c)+o(1)\). Projecting \(y\) onto
\([-\sqrt{c/(1+c)},\sqrt{c/(1+c)}]\) changes it by only \(o(1)\), so
we may assume that \(y\) lies in this interval without changing the
preceding estimates.

It remains to compose the retained outer and inner covers. Let
\[
X:=\{-x_j:j\in\widetilde S^{(\ell)}\},
\]
where repeated outer directions are included only once. For each
\(x\in X\), choose a representative \(j(x)\in\widetilde S^{(\ell)}\)
such that \(x=-x_{j(x)}\), and define
\[
Y_x:=\{r_{j(x)k}:k\in N_{j(x)}(J)\}.
\]
The endpoints generated by these representative families form a subset of
\(S^{(\ell+1)}\); retaining all endpoints can only increase the maximum in the final cover event.

We now apply Lemma~\ref{lem:orthogonal-cover-composition} with $E=5+3c$ and $U^2=1-y^2.$ To verify its hypotheses, note that
\(r_{jk}\perp\operatorname{span}\{v_i,v_j\}\), and hence
\(Y_x\subseteq x^\perp\). Furthermore, the family \(X\) is a
\(\bigl(\lambda t-o(t),\exp(-o(t^2))\bigr)\)-cover and uniformly
over \(x\in X\), the family \(Y_x\) is a
\(\bigl(\sqrt3\,t-o(t),\exp(-o(t^2))\bigr)\)-cover as \(\rho_t/48=\exp(-o(t^2))\). Lastly, the required cardinality and norm bounds follow from
\begin{align*}
\log|X|
&\le \log|\widetilde S^{(\ell)}|
 \le \frac{5+3c}{2}t^2+o(t^2),\\
\log|Y_x|
&\le \log\Delta(G)
 \le \frac{3(1+c)}{2}t^2+o(t^2),\\
\sup_{r\in Y_x}\|r\|^2
&\le 1-y^2+o(1),
\end{align*}
and the remaining parameter constraints follow from
\begin{align*}
E-\lambda^2
&=5+3c-\lambda^2\ge0,\\
U^2-\frac1{1+c}
&=\frac{c}{1+c}-y^2\ge0.
\end{align*}
Thus, all the hypotheses of Lemma~\ref{lem:orthogonal-cover-composition} are satisfied, so applying the lemma gives, with probability at least
\(\exp(-o(t^2))\), a representative \(j\) and an endpoint
\(k\in N_j(J)\) such that
\begin{align*}
\langle\gamma,-x_j\rangle
&\ge\lambda t-o(t),\\
\langle\gamma,r_{jk}\rangle
&\ge
\left(
\sqrt3-
\sqrt{
(5+3c-\lambda^2)
\left(\frac{c}{1+c}-y^2\right)
}
\right)t-o(t).
\end{align*}
We now verify the needed nondegeneracy bounds:
\begin{align*}
A(a,y)
&\ge
\frac12\sqrt{1-L^2}
-\frac{\sqrt3}{2}L\sqrt{\frac{c}{1+c}}
\ge\nu_0(c)>0,\\
|B(a,y)|^2
&\le\frac14+\frac34y^2
\le L^2<1.
\end{align*}
Therefore, writing
\[
x_k:=
\frac{v_k-\langle v_i,v_k\rangle v_i}
{\sqrt{1-\langle v_i,v_k\rangle^2}},
\]
the identities above give
\begin{align*}
\langle\gamma,x_k\rangle
&=
\frac{
\bigl(A(a,y)+o(1)\bigr)\langle\gamma,-x_j\rangle
+\frac{\sqrt3}{2}\langle\gamma,r_{jk}\rangle
}{
\sqrt{1-\langle v_i,v_k\rangle^2}
}\\
&\ge
\frac{
\bigl(A(a,y)+o(1)\bigr)\bigl(\lambda t-o(t)\bigr)
+\frac{\sqrt3}{2}
\left[
\left(
\sqrt3-
\sqrt{
(5+3c-\lambda^2)
\left(\frac{c}{1+c}-y^2\right)
}
\right)t-o(t)
\right]
}{
\sqrt{1-B(a,y)^2}+o(1)
}\\
&=
\frac{
\lambda A(a,y)+\frac32
-\frac{\sqrt3}{2}
\sqrt{
(5+3c-\lambda^2)
\left(\frac{c}{1+c}-y^2\right)
}
}{
\sqrt{1-B(a,y)^2}
}\,t-o(t)\\
&=
(\lambda_{\mathrm{new}}+\tau)t-o(t).
\end{align*}
Lowering the threshold by \(\tau t\), we conclude that
\[
\Pr_\gamma\!\left[
\max_{k\in S^{(\ell+1)}}
\langle\gamma,x_k\rangle
\ge\lambda_{\mathrm{new}}t-o(t)
\right]
\ge\exp(-o(t^2)).
\]
Together with
\(\langle v_i,v_k\rangle=B(a,y)+o(1)=a_{\mathrm{new}}+o(1)\),
this shows that \(S^{(\ell+1)}\) is the claimed recursive Gaussian
certificate.
\end{proof}

\section{Proof of Theorem \ref{thm:main}}
\label{subsec:dense-sparse-balance}

We combine Theorem \ref{thm:sparse-progress} with standard sparse--dense progress results from the literature~\cite{Blu94, KT12, KT17, BHL26}.

Blum's coloring framework~\cite{Blu94} states that an algorithm makes progress toward an
\(f(n)\)-coloring if it produces at least one of the following:
\begin{itemize}
\item (P1) An independent set of size \(\Omega(n/f(n))\)
\item (P2) A nonempty independent set \(S\) satisfying
      \(|N(S)|=O(f(n)|S|)\)
\item (P3) Two distinct vertices that receive the same color in every proper \(3\)-coloring.
\end{itemize}
Theorem~\ref{thm:sparse-progress} gives progress of type (P1). For the other two types of progress, we use the following two known results.

\begin{theorem}[Dense progress {\cite[Theorem~1]{KTY24}}]
\label{thm:dense-progress}
Let \(G\) be a \(3\)-colorable \(n\)-vertex graph of minimum degree
\(d>\sqrt n\). One can make progress in polynomial time toward an
\(n^{o(1)}\sqrt{n/d}\)-coloring.
\end{theorem}

\begin{theorem}[Progress combination {\cite[Proposition~17]{KT17}}]
\label{thm:progress-combination}
Let \(\alpha,\beta\in(0,1)\). Suppose that one can make progress in
polynomial time toward a \(\widetilde O(n^\alpha)\)-coloring whenever either
\(\Delta(G)\le n^\beta\) or \(\delta(G)\ge n^\beta\).
Then one can find a \(\widetilde O(n^\alpha)\)-coloring in polynomial time.
\end{theorem}

It remains to match the parameters.

\begin{proof}[Proof of Theorem~\ref{thm:main}]
Fix \(0<c<c^\dagger\), and put \(\beta(c):=3(1+c)/(5+3c)\) and
\(\alpha(c):=1/(5+3c)\). By Theorem~\ref{thm:sparse-progress}, when
\(\Delta(G)\le n^{\beta(c)}\), we make progress of type (P1) toward an
\(n^{\alpha(c)}\)-coloring. When \(\delta(G)\ge n^{\beta(c)}\),
Theorem~\ref{thm:dense-progress} makes progress toward
\(n^{(1-\beta(c))/2+o(1)}=n^{\alpha(c)+o(1)}\) colors, since
\begin{align*}
\frac{1-\beta(c)}2
&=\frac{1-\frac{3(1+c)}{5+3c}}2
=\frac1{5+3c}=\alpha(c).
\end{align*}
Given \(\varepsilon>0\), choose \(c<c^\dagger\) and a fixed exponent
\(\theta\) such that \(\alpha(c)<\theta<(13-\sqrt{97})/18+\varepsilon\).
This is possible because, as \(c\uparrow c^\dagger\),
\(\alpha(c)\longrightarrow1/(5+3c^\dagger)=(13-\sqrt{97})/18\).
For all sufficiently large \(n\), the sparse and dense routines both make
progress toward an \(n^\theta\)-coloring. Applying
Theorem~\ref{thm:progress-combination} gives
\begin{align*}
\widetilde O(n^\theta)
&=O_\varepsilon\!\left(n^{(13-\sqrt{97})/18+\varepsilon}\right)
\end{align*}
colors, since the fixed gap between \(\theta\) and the final exponent
absorbs the polylogarithmic factor.
\end{proof}

\end{document}